\documentclass[11pt]{article}
\usepackage[utf8x]{inputenc}
\usepackage[OT2,T1]{fontenc}
\usepackage{enumitem}
\makeatletter
\newcommand{\mylabel}[2]{#2\def\@currentlabel{#2}\label{#1}}
\makeatother
\usepackage{braket}
\usepackage{amsfonts,amssymb}
\usepackage{color}
\usepackage{amsthm, amssymb, latexsym, amsmath, dsfont
}
\usepackage[margin=1.2 in]{geometry}
\usepackage[ngerman,english]{babel}
\usepackage[utf8x]{inputenc}
\usepackage[OT2,T1]{fontenc}
\usepackage{mathabx}
\usepackage{booktabs}
\usepackage{siunitx}

\usepackage{hyperref}
\makeatletter
\def\eqref{\@ifstar\@eqref\@@eqref}
\def\@eqref#1{\textup{\tagform@{\ref*{#1}}}}
\def\@@eqref#1{\textup{\tagform@{\ref{#1}}}}
\makeatother

\newtheorem{theorem}{Theorem}[section]

\newtheorem{lemma}[theorem]{Lemma}

\newtheorem{proposition}[theorem]{Proposition}
\newtheorem{corollary}[theorem]{Corollary}

\theoremstyle{definition}

\newtheorem{remark}{Remark}

\newcommand{\DETAILS}[1]{}

\renewcommand{\Im}{\mathrm{Im}}

\newcommand{\E}{E}

\newcommand{\C}{\mathbb{C}}
\newcommand{\R}{\mathbb{R}}

\newcommand{\N}{\mathbb{N}}

\newcommand{\D}{{\cal D}}
\newcommand{\cA}{{\cal{A}}}

\newcommand{\cD}{{\cal{D}}}

\newcommand{\cH}{{\cal{H}}}

\newcommand{\cO}{{\cal{O}}}

\renewcommand{\i}{\mathrm{i}}

\newcommand{\vphi}{{\varphi}}
\newcommand{\e}{{\mathrm{e}}}
\newcommand{\vep}{{{\varepsilon}}}           

\newcommand{\scp}[2]{\left\langle #1\text{,}\, #2\right\rangle}
\newcommand{\tr}[1]{\mathrm{Tr}#1}

\begin{document}
	
	\title{Convergence rate towards the fractional Hartree-equation with singular potentials in higher Sobolev norms}
	
	\author{Michael Hott\thanks{ Department of Mathematics,
			University of Texas at Austin, 2515 Speedway,
			Austin, TX 78712, USA.}}
	
	\bigskip
	
	\bigskip
	
	\bigskip

	\bigskip
	\bigskip

	\maketitle
	
	\begin{abstract}
	   This is a work extending the results of \cite{AH} and \cite{AHH}. We want to show convergence of the Schr\"odinger equation towards the Hartree equation with more natural assumptions. We first consider both the defocusing and the focusing semi-relativistic Hartree equation. We show that the tools of \cite{P} are essentially sufficient for deriving the Hartree equation in those cases. Next, we extend this result to the case of fractional Hartree equations with possibly more singular potentials than the Coulomb potential. Finally, we show that, in the non-relativistic case, one can derive the Hartree equation assuming only $L^2$-data in the case of potentials that are more than or as regular as the Coulomb potential. We also derive the Hartree equation for more singular potentials in this case. This work is inspired by talks given at the conference 'MCQM 2018' at Sapienza/Rome.
	\end{abstract}
	\tableofcontents
\section{Introduction}

In 1924/25, Bose \cite{B} and Einstein \cite{E} predicted that if one adiabatically cools down a bosonic gas below a certain threshold, its ground state will become gradually more populated. 
This in turn has significant consequences on the partition function of the gas, which itself determines the macroscopic state of the gas. This effect is nowadays known as \textit{Bose-Einstein condensation}. Only in 1995 the group of Cornell and Wieman \cite{CW} and independently the group of Ketterle \cite{K} could experimentally verify this phenomenon, which then resulted in a Physics Nobel Prize in 2001. The first time this problem was mathematically rigorously treated was in 1980 in a work of H. Spohn \cite{S}, where persistence of the condensate was shown. Ever since the mathematical community became very interested in the problem and many improvements and extensions have been studied. Different techniques have been applied, ranging from coherent state analysis, see \cite{GV}, \cite{Hep}, to the BBGKY hierarchy, see \cite{S}, \cite{EY}, to a Egorov type theorem, see \cite{FGS}, \cite{FKP}, \cite{FKS}, Wigner measure approach, see \cite{AN}, \cite{AFP}, second quantization formulation, see \cite{RS}, \cite{MS}, \cite{Lu}, and deviation estimates \cite{KP}, \cite{P}. Whereas all the previous deal with the case of a dense gas of very weakly interacting particles, the somewhat dual and physically more important case of dilute gases with strong, delta type interactions has been treated in \cite{CH1}, \cite{CH2}, \cite{CP1}, \cite{CP2}, \cite{ESY1}, \cite{ESY2}, \cite{KM}, and references therein. In addition to all these results showing persistence of the condensate, there are also work establishing the formation of a condensate, see \cite{DSS}, \cite{LSY}, \cite{LSSY} and references therein. In the present work, we want to show persistence of the condensate in the mean-field limit, corresponding to the case of dense weakly acting gases.
\par Let us give some physical justification for the study of the fractional Schr\"odinger equation. In the case of the half-wave equation with kinetic energy operator $|\nabla|$, the equation describes a Boson star like system. A blow-up of the equation corresponds to a collapse of the Boson star. The blow-up in the effective equation has been shown \cite{Le}. In order to show that the whole system indeed exhibits the blow-up behavior, one has to compare the full Schr\"odinger dynamics with that of the effective equation. That was first successfully shown  in \cite{MS} and \cite{Lee}. More generally, fractional Quantum mechanics arises as a generalization of Brownian motion of Quantum particles as seen, e.g., in the Feynman path integral formalism as first studied by Laskin, see \cite{La1}, \cite{La2}. This more general class of (stochastic) jump processes is known as L\'evy processes. Several mathematical aspects of the corresponding fractional NLS, describing the condensate, have been studied in various scenarios, see \cite{CHHO}, \cite{CHKL}, \cite{D}, \cite{Hw}, \cite{HS}, \cite{PS}, \cite{Zh}, \cite{ZZ}.
\paragraph{The Setup} The Hamiltonian for a weakly interacting system of $N$ bosons is given by
\begin{equation}\label{manybodyham}
\begin{split}
H_{N}\;:=\;&\sum_{i=1}^{N}S_i + \frac{\mu\lambda}{N-1}\sum_{i<j}\frac1{|x_i-x_j|^\gamma}
\end{split}
\end{equation}
acting on $\cH_N\;:=\;L^2(\R^3)^{\otimes_S N}$. In here, $\mu=\pm1$ distinguishes the focusing/defocusing case and $\lambda>0$ is some coupling strength. The kinetic energy operator is given by
\begin{equation*}
S_i=(-\Delta_{x_i})^\sigma.
\end{equation*}
Physically relevant cases, that we want to analyze here, are the non-relativistic case, $(\gamma,\sigma)=(1,1)$, and the semi-relativistic case, $(\gamma,\sigma)=(1,1/2)$. In addition to these, we also want to consider the cases $\gamma\in(0,3/2)$, $\sigma\in[\gamma/2,1]\setminus\{1/2\}$. In the following, we consider a solution $\Psi_{N,t}$ of the Cauchy problem
\begin{equation}\label{schroedinger}
\begin{cases}
\i\partial_t\Psi_{N,t}\;&=\;H_N\Psi_{N,t}\\
\Psi_{N,t}\;&=\;\vphi_0^{\otimes N}
\end{cases}	
\end{equation}
for given $\vphi_0\in H^s(\R^3)$ for some given $s\geq0$ with $\|\vphi_0\|_2=1$, where as usual we set $H^0(\R^3):=L^2(\R^3)$. Moreover, let $\vphi_t$ be a solution of the Hartree equation given by
\begin{equation}\label{hartree}
\begin{cases}
\i\partial_t\vphi_t\;&=\;S \vphi_t+\frac{\mu\lambda}{|\cdot|^\gamma}*|\vphi_t|^2\vphi_t\\
\vphi_{t=0}\;&=\;\vphi_0,
\end{cases}
\end{equation}
with $S$ chosen accordingly to $S_i$. Since, in the focusing case, $\mu=-1$, $H_N$ is a self-adjoint and semi-bounded operator only for certain regions of the parameter $(\lambda,\gamma,\sigma)$, as we will see below, let us introduce a regularized Hamiltonian 
\begin{equation}\label{regham}
H_{N}^{(\alpha)}\;=\;\sum_{i=1}^{N}S_i + \frac{\mu\lambda}{N-1}\sum_{i<j}\frac1{|x_i-x_j|^\gamma+\alpha},\quad\alpha>0.
\end{equation}
Note that $H_N^{(\alpha)}$ is both self-adjoint and bounded from below. In fact, we have
\begin{equation}\label{reghamboundskin}
H_N^{(\alpha)}\geq\sum_{i=1}^{N}S_i +\frac{N\mu\lambda}{2\alpha}\geq\frac{N\mu\lambda}{2\alpha}.
\end{equation}
Thus we may write down the formal solution 
\begin{equation}\label{regschroedinger}
\Psi_{N,t}^{(\alpha)}\;:=\;\e^{-\i H_N^{(\alpha)}t}\vphi_0^{\otimes N}
\end{equation}
to the respective Cauchy problem. Let us also introduce the regularized Hartree equation given by
\begin{equation}\label{reghartree}
\begin{cases}
i\partial_t\vphi_t^{(\alpha)}\;&=\;S \vphi_t^{(\alpha)}+\frac{\mu\lambda}{|\cdot|^\gamma+\alpha}*|\vphi_t^{(\alpha)}|^2\vphi_t^{(\alpha)}\\
\vphi_{t=0}^{(\alpha)}\;&=\;\vphi_0.
\end{cases}
\end{equation}
\par In the following, our goal is to show convergence of \eqref{schroedinger} towards \eqref{hartree} in the sense of expectation values, i.e., we show
\begin{equation}\label{convtype}
|\scp{\Psi_{N,t}}{\mathcal{A}\otimes I_{N-k}\Psi_{N,t}}-\scp{\vphi_t^{\otimes k}}{\mathcal{A}\vphi_t^{\otimes k}}|\lesssim_{t,k,\mathcal A, \vphi} o(1)
\end{equation}
for a class of operators $\mathcal{A}:L^2(\R^{3k})\circlearrowleft$ such that the involved bounding constant is finite.
\paragraph{Assumptions:} For the subsequent analysis, we will distinguish the following cases:
\begin{itemize}
	\item[\mylabel{itm:sr}{{\rm (}SR{\rm )}}] $\gamma=1$, $\sigma=1/2$.
	\item[\mylabel{itm:fs}{{\rm (}FS{\rm )}}] $\gamma\in(0,3/2)$, $\sigma\in[\gamma/2,1)\setminus\{1/2\}$.
	\item[\mylabel{itm:nr}{{\rm (}NR{\rm )}}] $\gamma\in(0,3/2)$, $\sigma=1$.
\end{itemize}
If not specified, we will generally assume that $\gamma\in(0,3/2)$, $\sigma\in[\gamma/2,1]$.

\paragraph{The heuristics.} Before turning to the actual derivation of \eqref{hartree} from \eqref{schroedinger}, let us give us give a heuristic argument. Assuming full condensation at positive times $\Psi_{N,t}=\vphi_t^{\otimes N}$, we obtain, using \eqref{schroedinger},
\begin{align*}
\scp{\vphi_t}{\i\partial_t \vphi_t}&=\frac1N\scp{\Psi_{N,t}}{\i\partial_t\Psi_{N,t}}=\frac1N\scp{\Psi_{N,t}}{H_N\Psi_{N,t}}\\
&=\scp{\vphi_t}{S\vphi_t}+\frac{\mu\lambda}{2}\scp{\vphi_t^{\otimes 2}}{\frac1{|x_1-x_2|^\gamma}\vphi_t^{\otimes 2}}\\	
&=: h(\vphi_t,\overline{\vphi_t}).
\end{align*}
$h(\vphi_t,\overline{\vphi_t})$ can be understood as a Hamiltonian function from which we obtain the dynamics
\begin{equation}\label{hartreederived}
\begin{cases}
\i\partial_t\vphi_t\;&=\;\partial_{\overline{\vphi_t} }h(\vphi_t,\overline{\vphi_t})=S\vphi_t+\frac{\mu\lambda}{|\cdot|^\gamma}*|\vphi_t|^2\vphi_t\\
\vphi_t\big|_{t=0}\;&=\;\vphi_0.
\end{cases}
\end{equation}
\paragraph{The Hamiltonian structure.} The Hamiltonian dynamics in \eqref{schroedinger} is induced by the Hamiltonian $H_N(\Psi):=\scp{\Psi}{H_N\Psi}$ and the symplectic form $$\omega_N:\cH_N\times\cH_N\to\R,(\Psi,\Phi)\mapsto-2\Im\scp{\Psi}{\Phi}.$$Indeed, if $\Psi_N$ satisfies \eqref{schroedinger}, we have
$$\omega_N(\Phi,\partial_t\Psi_{N,t})=dH_N(\Psi_{N,t})[\Phi]\quad\forall \Phi\in\cH_N,$$ 
where $dH_N$ denotes the first variation of $H_N$, i.e., 
$$dH_N(\Psi_{N,t})[\Phi]=\partial_\varepsilon\big|_{\varepsilon=0}H_N(\Psi_{N,t}+\varepsilon\Phi).$$
We consider the restriction of $\omega_N$ to the submanifold 
$$M:=\{\vphi^{\otimes N}\mid \|\vphi\|_2=1\}\subseteq \cH_N$$
Its tangent space is given by
$$T_\vphi M=\sum_{j=1}^N\vphi^{\otimes (j-1)}\otimes \{\vphi\}^{\perp}\otimes \vphi^{\otimes(N-j)}.$$
Then the restriction of $\omega_N$ to $T_\vphi M$ is given by
$$\omega_N\big|_{T_\vphi M\times T_\vphi M}\;=\;-2N\Im\scp{\cdot}{\cdot}\big|_{L^2(\R^3;S^1)\times L^2(\R^3;S^1)}\;=:\;N\omega,$$
where $S^1\subset \C$ is the unit circle. Then $\omega$ together with $h(\overline{\vphi},\vphi)$ induces the dynamics \eqref{hartree}. Indeed, one can check that a solution $\vphi_t$ of \eqref{hartree} satisfies
$$\omega(\psi,\partial_t\vphi_t)=dh(\vphi_t)[\psi]\quad \forall \psi\in \{\vphi_t\}^\perp.$$
\paragraph{On the scaling.} Before we turn to discussions related to the scaling of \eqref{hartree}, let us first see how \eqref{hartree} is related to its usually considered form 
\begin{equation}\label{hartreerescaled}
\begin{cases}
\i \partial_t u_t\;&=\;S u_t+\frac{\mu}{|\cdot|^\gamma}*|u_t|^2u_t\\
u_t\big|_{t=0}\;&=\;u_0.
\end{cases}	
\end{equation}
The relation between \eqref{hartreerescaled} and \eqref{hartree} is given by $u=\sqrt{\lambda}\vphi$. In particular, we have due to mass conservation, see Lemma \ref{fracconservation} below, and normalization of $\vphi_0$
\begin{equation}\label{lambdascale}
\lambda=\|u\|_2^2.
\end{equation}
This will become important when we discuss well-posedness of \eqref{hartree} below. We will give the discussion depending on $\lambda$ which by the previous identity can be understood as size of initial datum for \eqref{hartreerescaled}.
\par The scaling, under which \eqref{hartree} remains invariant, is given by
$$\vphi(x,t)\rightarrow \ell^{\frac{d-\gamma}{2}+\sigma}\vphi(\ell x,\ell^{2\sigma}t).$$
This scaling leaves the $\dot{H}^{s_c}(\R^3)$-norm invariant where $s_c=\frac{\gamma}{2}-\sigma$ is the critical exponent. Note that $s_c<0$ corresponds to $\sigma>\gamma/2$ with mass criticality, $s_c=0$, at $\sigma=\gamma/2$. In the present work, we will only work with the mass subcritical regime $\sigma>\gamma/2$ and the mass critical case $\sigma=\gamma/2$. Also, instead of the usual form \eqref{hartreerescaled} of the Hartree equation, we will see that for the sake of its derivation it is more convenient to work with the form \eqref{hartree}.

\paragraph{Acknowledgment} The author is grateful to numerous helpful discussions with and suggestions from T. Chen and N. Pavlovi\'{c}. The author also wants to thank M. Rosenzweig for helpful discussions. This work was funded through the 'University of Texas at Austin Provost Graduate Excellence Fellowship'.

\section{Results}
Results of the form \eqref{convtype} go back to the 80's where H. Spohn \cite{S} showed this result in the case of bounded operators $\mathcal{A}$ with the potential $|\cdot|^{-\gamma}$ replaced by a bounded potential. After that, several parameters have been optimized, from considering the Coulomb case \cite{EY} to obtaining explicit rates \cite{RS} to even considering unbounded operators $\mathcal{A}$ \cite{MS}. The methods in these were quite involved, usually using either hierarchies or second quantization methods from QFT. Then, P. Pickl (\cite{P}, \cite{KP}) suggested to consider the projection of the solution $\Psi_{N,t}$ of \eqref{schroedinger} onto the orthogonal complement of the solution $\vphi_t$ of \eqref{hartree}. Using simple algebra with projectors onto the span of $\vphi_t$, he was able to show \eqref{convtype} for bounded operators $\cA$. In \cite{AH} and \cite{AHH}, it was then shown how to extend this result to unbounded operators $\cA$. In the present work, we will see that the methods of \cite{P} and \cite{AH} are sufficient to both prove most of the previous with a much simpler proof and to even improve rates of convergence as well as decrease required regularities. In order to give an overview of the previous results and in order to state our main result, let us introduce some parameters. Let $s$ denote the regularity on the initial data $\vphi_0\in H^s(\R^3)$, $\beta\geq0$ be the rate of convergence $N^{-\beta}$ in \eqref{convtype}, where $\beta=0$ stands for $o(1)$-convergence. Next, let $k\in\N$ be the number of particles on which $\cA$ acts and $\theta\geq0$ be its "degree of unboundedness", i.e., let $k$ and $\theta$ be such $S_{k,1}^{-\theta/2}\cA S_{k,1}^{-\theta/2}$ can be extended to a bounded linear operator $L^2\to L^2$. In here, $S_{k,r}:=\sum_{i=1}^k (1+S_i)^r$ abbreviates a differential operator. Let $f(t)$ denote the growth of the error-bound in \eqref{convtype} as a function of $t$. With $-$ we denote that only a uniform bound on a compact time interval $[0,T]$ has been derived. The appearing constants $C$ generally depend on all the parameters $\sigma,\mu,\gamma,s,\beta$ as well on the size $\|\vphi_0\|_{H^s}$ of the initial datum. The subscript $k$ denotes that there is an implicit dependence on the involved particle number $k$. We remark that \cite{P} gives an explicit dependence of $C_k, D_k$ on $k$. Then we can summarize some of the important results in table \ref{previousresultstab}. 
\par Instead of a bound of the type \eqref{convtype}, all of the listed results rather bound the Sobolev trace norm
$$\tr{\left|S_{k,r}^{\frac12}\left(\gamma_{N,t}^{(k)} -P^{(k)}_t\right)S_{k,r}^{\frac12}\right|}.$$
Then, using the definition of the reduced density matrix and the cyclicity of the trace, see \cite{AHH}, we have  
\begin{equation}\label{tracecyclicity}
\begin{split}
\text{l.h.s of } \eqref{convtype} &= \left|\tr\Big( \cA\big(\gamma_{N,t}^{(k)} -P^{(k)}_t\big) \Big) \right|
= \left|\tr\Big( S_{k,r}^{-\frac12}\cA S_{k,r}^{-\frac12} S_{k,r}^{\frac12}\big(\gamma_{N,t}^{(k)} -P^{(k)}_t\big) S_{k,r}^{\frac12}\Big) \right| \\
&\le \left\|S_{k,r}^{-\frac12}\cA S_{k,r}^{-\frac12}\right\|_{L^2\to L^2} 
\tr{\left|S_{k,r}^{\frac12}\left(\gamma_{N,t}^{(k)} -P^{(k)}_t\right)S_{k,r}^{\frac12}\right|}
\end{split}
\end{equation}
since the space of bounded operators is the dual of the trace class operators, see \cite[Theorem VI.26]{RS1}. Thus, regarding optimal regularity, the two quantities to compare are the imposed regularity $s$ on $\vphi_0$ and the needed regularity $\theta\sigma$ in order to formulate the problem respectively bounding the r.h.s. of  \eqref{tracecyclicity}. In the optimal case, these coincide.

\begin{table}
	\centering
	\begin{tabular}{lcccccccc}\toprule
		& $\sigma$ & $\mu$ & $\gamma$ & $s$ & $\beta$ & $k$ & $\theta$ & $f(t)$ \\
		\hline
		\cite{AHH} & $1/2$ & $1$ & 1 & 1 & $(1-\theta)/2$ & $\N$ & $[0,1)$ & $C\e^{C\e^{Ct}}$\\ 
		\cite{Lee} & $1/2$ & $\pm1$ & 1 & 1 & 1 & 1 & 0 & $f(t)$\\
		\cite{AHH} + \cite{Lee} & $1/2$ & $1$ & 1 & 1 & $\min\{1/2;1-\theta\}$ & 1 & [0,1)$^{(*)}$ & $f(t)$\\
		\cite{MS} & $1/2$ & $-1$ & 1 & 2 & 1/2 & 1 & 0 & $-$\\
		& $1/2$ & $-1$ & 1 & 2 & 1/4 & 1 & [0,1] & $-$\\ \midrule
		\cite{EY} & $1$  & $\pm1$ & $1$ & $2$ & $0$ & $\N$ & $0$ & $-$\\
		\cite{RS}, \cite{P} & 1 & $\pm1$ & 1 & 1 & 1/2 & 1 & 0 & $C_k\e^{D_kt}$\\
		\cite{KP} & $1$ & $\pm1$ & $(0,\frac32)$ & $1$ & $1/2$ & $\N$ & 0 & $C\e^{Ct}$\\
		\cite{CLS} & 1 & $\pm1$ & 1 & 1 & 1 & $\N$ & 0 & $C_k\e^{D_kt}$\\
		\cite{FKS} & 1 & $\pm1$ & $(0,1]$ & 0 & $>0^{(\dagger)}$ & $\N$ & 0 & $f(t)$\\
		\cite{Lu} & 1 & $\pm1$ & 1 & 1 & $1/2$ & $\N$ & 0 & $C\e^{Ct}$\\  
		 & 1 & $\pm1$ & 1 & 3 & 1/4 & 1 & [0,1] & $f(t)$\\
		\cite{AH} & 1 & $\pm1$ & 1 & 1 & $\min\{1/2;1-\theta\}$ & $\N$ & [0,1) & $C_k\e^{D_kt}$\\
		& 1 & $\pm1$ & 1 & 1 & 0 & $\N$ & [0,1] & $-$\\
		\cite{CLL} & 1 & $\pm1$ & $(0,\frac32)$ & 1 & 1 & 1 &0 & $C\e^{Ct^{3/2}}$\\ \bottomrule
	\end{tabular}
	\caption{Previous results. $(*)$ has not been showed explicitly in \cite{AHH} but can be easily obtained as a corollary together with \cite{Lee}. See also Proposition \ref{srlee}. $(\dagger)$ means that there is a rate $N^{-\beta}$ for some $\beta>0$. See also the remark \ref{remarkonresults}.}
	\label{previousresultstab}
\end{table}

An important remark on both \cite{Lee} and \cite{MS} is that they consider both defocusing ($\mu=1$) and the focusing case ($\mu=-1$). 
For a result in the focusing case, we need to replace $\Psi_{N,t}$ in \eqref{convtype} by $\exp(-\i H_N^{(\alpha)}t)\Psi_0$, see \eqref{regham}. 
\par Note, that there are results including further parameters like a magnetic potential $A$, see \cite{Lu}, or the number of species $r$ involved in the condensate, see \cite{AHH}, \cite{Hei}, and \cite{MO}. 
As we will also see below, whenever $k=1$ and $\theta=0$, one can show that one can obtain a result for any $k\in\N$ if one sacrifices half of the rate of convergence. We are able to state our main result.

\begin{theorem}\label{main}
	Let $(\sigma,\mu,\gamma,s,k,\theta)$ be given as in a line of table \ref{resulttable}. Assume $\vphi_0\in H^s(\R^3)$ and $N\in\N^{\geq k}$. Suppose $\Psi_{N,t}$ is a solution of \eqref{schroedinger}, $\vphi_t$ is a solution of \eqref{hartree}, and $\Psi_{N,t}^{(\alpha)}=\e^{-\i H_N^{(\alpha)}t}\vphi_0^{\otimes N}$. Let $\cA$ be a self-adjoint operator acting on $L^2(\R^3)^{\otimes_S k}$. Assume that $S_{k,1}^{-\theta/2} \cA S_{k,1}^{-\theta/2}$ can be extended to a bounded operator on $L^2(\R^3)^{\otimes_S k}$ with operator norm $\|S_{k,1}^{-\theta/2} \cA S_{k,1}^{-\theta/2}\|$. Then we have the following.
	\begin{enumerate}
		\item In the defocusing case, $\mu=1$, there is a constant $C=C_{\|\vphi_0\|_{H^{s}}}$ and a function $f:\R\to\R$ such that for any $t\geq0$ we have
		$$\left|\scp{\Psi_{N,t}}{\cA\otimes I_{N-k} \Psi_{N,t}}-\scp{\vphi_t^{\otimes k}}{\cA\vphi_t^{\otimes k}}\right| \leq C\|S_{k,1}^{-\theta/2} \cA S_{k,1}^{-\theta/2}\|\frac{k^{\frac{3-\theta}2}f(t)}{N^\beta},$$
		where $\beta>0$ and $f(t)$ are the values in the line in table \ref{resulttable} corresponding to the chosen parameter $(\sigma,\mu,\gamma,s,k,\theta)$.
		\item In the focusing case, $\mu=-1$, fix $T\in(0,\infty)$ be such that 
		$$\nu\;=\;\sup_{|\tau|\leq T} \|\vphi_\tau\|_{H^{\frac\gamma2}}<\infty.$$ Then there is a constant $C=C_{\nu,\|\vphi_0\|_{H^s},T}$ and a sequence $(\alpha_n)_{n\in\N}\in(0,\infty)^\N$ such that for any $t\in[0,T)$ we have
		$$\left|\scp{\Psi_{N,t}^{(\alpha_N)}}{\cA\otimes I_{N-k} \Psi_{N,t}^{(\alpha_N)}}-\scp{\vphi_t^{\otimes k}}{\cA\vphi_t^{\otimes k}}\right|\leq C\|S_{k,1}^{-\theta/2} \cA S_{k,1}^{-\theta/2}\|\frac{k^{\frac{3-\theta}2}}{N^{\beta}},$$	
		where $\beta>0$ is the value in the line in table \ref{resulttable} corresponding to the chosen parameter $(\sigma,\mu,\gamma,s,k,\theta)$.
	\end{enumerate}
	
	
\end{theorem}
	\begin{table}[htp]
	\centering
	\begin{tabular}{cccccccc}\toprule
		$\sigma(\geq\gamma/2)$ & $\mu$ & $\gamma$ & $s$ & $\beta$ & $k$ & $\theta$ & $f(t)$\\
		\hline
		$1/2$ & $1$ & $1$ & $2/3$ & $1/2$ & $\N$ & $0$ & $\e^{C\e^{Ct}}$\\
		$1/2$& $1$ &  $1$ & $2/3$ & $(1-\theta)/2$ & $\N$ & $[0,1)$ & $\e^{C\e^{Ct}}$\\
		$1/2$ & $1$ & $1$ & $1$ & $\min\{1/2;1-\theta\}$ & $1$ & $[0,1)$ & $f(t)$\\
		$1/2$& $-1$ & $1$ & $2/3$ & $1/2$ & $\N$ & $0$ & $-$\\
		$1/2$& $-1$ & $1$ & $\frac{(1+\vep)}2^{(*)}$ & $\frac{(1-\theta)(1-\theta+\varepsilon(1+\theta)-\max\{1;2\theta\})}{2[1-\theta+\varepsilon(1+\theta)]}$ & $\N$ & $[0,\min\{\frac{\vep}{1-\vep};1\})$ & $-$\\
		$1/2$& $-1$ & $1$ & $1$& $\min\{1/2;1-\theta\}^2$ & $1$ & $[0,1)$ & $-$\\ \midrule
		$[\frac\gamma2,1)\setminus\{\frac12\}$ & $1$ & $(0,1)$ & $(1-\sigma)\gamma$ & $1/2$ & $\N$ & $0$ & $\e^{C\e^{Ct}}$\\
		$[\frac\gamma2,1)\setminus\{\frac12\}$& $1$ & $(0,1)$ & $(1-\sigma)\gamma$ & $(1-\theta)/2$ & $\N$ & $[0,1)$ & $C\e^{C\e^{Ct}}$\\
		$[\frac\gamma{\gamma+1},1)\setminus\{\frac12\}$& $1$ & $(0,\frac32)$ & $\sigma$ & $1/2$ & $\N$ & $0$ & $\e^{C(\sqrt{t}+t^2)}$\\
		$[\frac\gamma{\gamma+1},1)\setminus\{\frac12\}$& $1$ & $(0,\frac32)$ & $\sigma$ & $(1-\theta)/2$ & $\N$ & $[0,1)$ & $C\e^{C(\sqrt{t}+t^2)}$\\
		$\gamma/2$& $-1$ & $(0,1)$ & $(1-\gamma/2)\gamma$ & $1/2$ & $\N$ & $0$ & $-$\\
		$\gamma/2$ & $-1$ & $(0,1)$ & $\frac{(1+\vep)\gamma}{2}^{(\dagger)}$ & $\frac{(1-\theta)(1-\theta+\varepsilon(1+\theta)-\max\{1;2\theta\})}{2[1-\theta+\varepsilon(1+\theta)]}$ & $\N$ & $[0,\min\{\frac{\vep}{1-\vep};1\})$ & $-$\\
		$\gamma/2$& $-1$ & $(1,\frac32)$ & $\gamma/2$ & $1/2$ & $\N$ & $0$ & $-$\\
		$\gamma/2$& $-1$ & $(1,\frac32)$ & $\frac{(1+\vep)\gamma}{2}^{(\#)}$ & $\frac{(1-\theta)(1-\theta+\varepsilon(1+\theta)-\max\{1;2\theta\})}{2[1-\theta+\varepsilon(1+\theta)]}$ & $\N$ & $[0,\min\{\frac{\vep}{1-\vep};1\})$ & $-$\\ \midrule
		$1$& $\pm1$ & $(0,1]$ & $0$ & $1/2$ & $\N$ & $0$ & $\e^{C(\sqrt{t}+t^6)}$\\
		$1$& $\pm1$ & $(0,1]$ & $(0,1]$ & $(s-\theta)/2$ & $\N$ & $[0,s)$ & $C\e^{C(\sqrt{t}+t^6)}$\\
		$1$& $\pm1$ & $(1,\frac32)$ & $\gamma-1$ & $1/2$ & $\N$ & $0$ & $\e^{C(\sqrt{t}+t^6)}$\\
		$1$& $\pm1$ & $(1,\frac32)$ & $[\gamma-1,1]$ & $(s-\theta)/2$ & $\N$ & $[0,s)$ & $C\e^{C(\sqrt{t}+t^6)}$\\
		$1$& $\pm1$ & $(0,\frac32)$ & $1$ & $\min\{1/2;1-\theta\}$ & $1$ & $[0,1)$ & $C\e^{C(\sqrt{t}+t^6)}$\\ \bottomrule
	\end{tabular}	
	\caption{Present results. $(*):$ $\vep\geq1/3$. $(\dagger):$ $\varepsilon\geq (1-\gamma)$. $(\#):$ $\vep>0$.}
	\label{resulttable}
\end{table}
\begin{remark}
	None of these convergence rates are optimal. The optimal rate $\cO(N^{-1})$ was obtained, e.g., in \cite{CLS}, \cite{CLL}, \cite{ES}. To see why this is optimal, \cite{ES} shows that for $A$ and $B$ acting on different particles, then $[A,\e^{\i H_Nt}B\e^{-\i H_N t}]$ remains of order $1/N$ at positive times $t>0$. This is then used to show that the difference of reduced density matrices $\gamma_{N,t}^{(m+n)}-\gamma_{N,t}^{(m)}\otimes\gamma_{N,t}^{(n)}$, see below, tested against products of an $m$-particle operator and an $n$-particle operator remains of order $1/N$ at positive times $t>0$. Moreover, as pointed out to the author by M. Machedon at the conference \textit{'TexAMP 2017'}, in order to obtain optimal time dependency of the error, it is essential to include two-particle correlations in the effective field. For details, we refer to \cite{GM1}, \cite{GM2}, \cite{M}, and \cite{MPP}. 
\end{remark}
\begin{remark}\label{remarkonresults}
	We want to comment on the present results. Notice that, in the defocusing semi-relativistic case, we obtain convergence in all Sobolev trace norms below the energy trace norm only assuming $\vphi_0\in H^{3/4}(\R^3)$. This drastically improves the $H^2$-assumption on $\vphi_0$ given in \cite{MS}, while we though do not cover the boundary case of the energy trace norm, see remark \ref{optimalregrem} below. In addition, we prove convergence for every $k$-marginal, $k\in\N$, with explicit dependence on $k$. Another remarkable result is that in the non-relativistic case with $\gamma\leq1$, we only need to assume $L^2$-initial data to obtain convergence with rate $N^{-1/2}$ and explicit dependency of the error bound on the parameters $k$, $\lambda$, and on the time $t$. This improves the result given in \cite{FKS} in that they only obtain some rate $N^{-\beta}$ with $\beta>0$ with some error bound depending on $k$, $\lambda$, and on the time $t$. In both cases, we can simplify the arguments a lot.
\end{remark}
\par A comment on the subsequently used constants: Whenever a constant only depends on the parameters $\gamma,\sigma,\mu,\lambda, s,\theta$, we will call that constant \textit{universal}, for these constants remain finite regardless of the chosen parameters. Unless mentioned otherwise, all occurring constants will depend on these. Moreover, by abuse of notation, we will use the same notation for a constant that possibly changes its value along proofs. This will help us reduce notation and make the arguments clear to the reader.
\section{Preliminary tools and notations \label{notations}}
We introduce 
\begin{align*}
\E^{(\gamma,\sigma)}[u]\;&:=\;\frac12\|(-\Delta)^{\frac\sigma2}u\|_2^2+\frac14\scp{u}{\frac{\mu\lambda}{|\cdot|^\gamma}*|u|^2u},\\
T^{(\sigma)}[u]\;&:=\;\frac12\|(-\Delta)^{\frac\sigma2} u\|_2^2,\\
V^{(\gamma)}[u]\;&:=\;\frac14\scp{u}{\frac{\mu\lambda}{|\cdot|^\gamma}*|u|^2u}.
\end{align*}
For a proof of the subsequent statements in this section, see appendix \ref{fracpersproof}.
\begin{lemma}[Conservation laws]\label{fracconservation}
	Suppose $\gamma\in(0,3/2)$, $\sigma\in[\gamma/2,1]$, and $s\geq\sigma$. Let $\vphi$ be a solution of \eqref{hartree} in $H^s(\R^3)$ with initial value $\vphi_0\in H^s(\R^3) $. Then both, energy and $L^2$-mass of $\vphi$ are conserved, i.e., we have
	$$\E^{(\gamma,\sigma)}[\vphi_t]\;=\;\E^{(\gamma,\sigma)}[\vphi_0]\quad\mbox{and}\quad \|\vphi_t\|_2=\|\vphi_0\|_2.$$
\end{lemma}
\begin{lemma}[Positivity of energy]\label{energypos}
	Suppose $\gamma\in(0,3/2)$, $\sigma\in[\gamma/2,1]$, and let $u\in H^\sigma(\R^3)$. Then there is a constant $C=C(\gamma,\sigma)$ such that we have
	\begin{equation}\label{venergyestimate}
	|V^{(\gamma)}[u]|\;\leq\;\lambda CT^{(\sigma)}[u]^{\frac{2\gamma}{\sigma}}\|u\|_2^{4-\frac{\gamma}{\sigma}}.
	\end{equation}
	\par In particular, if $\|u\|_2=1$ and either $\sigma>\gamma/2$, or $\sigma=\gamma/2$ and $\mu\lambda C>-1$, we have $$T^{(\sigma)}[u]+1\lesssim\E^{(\gamma,\sigma)}[u]+1\lesssim T^{(\sigma)}[u]+1.$$
\end{lemma}
\begin{proof}
	Applying the Hardy-Littlewood-Sobolev inequality and then the Gagliardo-Nirenberg-Sobolev inequality, we find
	\begin{equation*}
	\begin{split}
	|V^{(\gamma)}[u]|&\lesssim\|u\|_{12/(6-\gamma)}^4\\
	&\lesssim\|(-\Delta)^{\sigma/2}u\|_2^{\frac{4\gamma}{\sigma}}\|u\|_2^{4-\frac{\gamma}{\sigma}}.
	\end{split}
	\end{equation*}

\end{proof}
\begin{lemma}[Self-adjointness of $H_N$]\label{hnsa}
	There is some $C=C(\gamma,\sigma)>0$, independent of $N$, such that, if either $\sigma>\gamma/2$, or $\sigma=\gamma/2$ and $\mu\lambda C>-1$, then $H_N$ is both self-adjoint and positive. In this case, we even have
	\begin{equation}\label{hamsupercrit}
	\sum_{i=1}^NS_i+N\lesssim H_N+N\lesssim\sum_{i=1}^NS_i+N
	\end{equation}
	in the sense of quadratic forms.
\end{lemma}

\begin{remark}
	Note that one can use the estimates for the proof of Lemma \ref{hnsa} in order to establish \eqref{venergyestimate}. In particular, if we choose $C$ optimal in \eqref{venergyestimate}, it will be below the constant $C(\gamma,\sigma)$ determined in Lemma \ref{hnsa}.
\end{remark}
When minimizing the constant in \eqref{venergyestimate}, we obtain the following soliton equation
\begin{equation}\label{groundstate}
(-\Delta)^{\gamma/2} Q-\frac{1}{|\cdot|^\gamma}*|Q|^2 Q=0.
\end{equation}
Let $Q\in H^{\gamma/2}(\R^3)$ denote the ground state solution. In view of \eqref{lambdascale} and due to scaling invariance of \eqref{groundstate}, we may define
\begin{equation*}
\lambda_{H,c}:=\begin{cases}\|Q\|_2^2\quad &\mbox{\ if\ }\sigma=\gamma/2\\
\infty\quad&\mbox{\ if\ }\sigma>\gamma/2.
\end{cases}
\end{equation*}
As pointed out in the last remark, if we take $C(\gamma,\sigma)$ as in Lemma \ref{hnsa} and define
\begin{equation*}
\lambda_{S,c}:=\begin{cases}1/C(\gamma,\gamma/2)\quad &\mbox{\ if\ }\sigma=\gamma/2\\
\infty\quad&\mbox{\ if\ }\sigma>\gamma/2,
\end{cases}
\end{equation*}
we have $\lambda_{S,c}\leq\lambda_{H,c}$. Note that $\lambda_{S,c}$ is only a lower bound on the optimal constant $\lambda_c^S(N)$ above which the Hamiltonian $H_N$ ceases to be semi-bounded. $\lambda_c^S(N)$ is the unique constant such that in the focusing case, $\mu=-1$, for all $\lambda<\lambda_c^S(N)$, $H_N$ is bounded from below and for $\lambda>\lambda_c^{S}(N)$ we have
$$\inf_{\|\Phi\|_2=1}\scp{\Phi}{H_N\Phi}=-\infty.$$
In the special case $\sigma=1/2=\gamma/2$ it was shown in \cite{LY} that for some universal constant $c>0$ we have
$$\lambda_{H,c}(1-cN^{-1/3})\leq \lambda_c^S(N)\leq \lambda_{H,c}(1+cN^{-1}).$$
As far as the author is concerned, there is no such result in the case of the general fractional Schr\"odinger equation. For the present work, the optimal value of $\lambda_{S,c}$ resp. $\lambda_c^S(N)$ is irrelevant because we want to focus on the derivation of the NLS instead.
 
\begin{proposition}[Well-posedness]\label{fracwellposed}
	Suppose $\gamma\in(0,3/2)$, $\sigma\in[\gamma/2,1]$, and $s\geq\sigma$. Then \eqref{hartree} is well-posed in $H^s(\R^3)$. More precisely, there is $T_{max}\in(0,\infty]$ such that for any $0<T<T_{max}$ there is a unique solution $\vphi \in C\left([0,T); H^s(\R^3)\right) \cap  C^1\left([0,T);  H^{s-2\sigma}(\R^3)\right)$. Furthermore, the solution $\vphi$ continuously depends on the initial datum $\vphi_0$. 
	\par Moreover, if $\mu=1$ or $\lambda<\lambda_{H,c}$, $T_{max}=\infty$, i.e., \eqref{hartree} is globally well-posed. If $\lambda>\lambda_{H,c}$ and $\mu=-1$, there is a family of initial data such that the corresponding solution of \eqref{hartree} blows up in finite time.
\end{proposition}
From the proof of this proposition, we get control over the growth rate of $H^s$-norms as follows.
\begin{lemma}[Persistence of regularity]\label{fracpersreg}
	Suppose $\gamma\in(0,3/2)$, $\sigma\in[\gamma/2,1]$, $s\geq0$, $r:=\max\{s;\sigma\}$, and $\vphi_0\in H^r(\R^3)$. Fix a time $T\in(0,\infty)$ such that
	$$\nu:=\nu(T)\;:=\;\sup_{|\tau|\leq T} \|\vphi_\tau\|_{H^{\frac\gamma2}}<\infty.$$
	Let $\vphi$ be a solution of \eqref{hartree} on $[0,T]$. Then there is a universal constant $c$ such that for all $t\in(0,T]$ we have
	$$\|\vphi_t\|_{H^s}\lesssim\e^{c\nu^2t} \|\vphi_0\|_{H^s}.$$
	If $\mu=1$ or $\lambda<\lambda_{H,c}$, we can chose $T=\infty$ and we even have for any $0<s\leq\sigma$
	$$\|\vphi_t\|_{H^s}\lesssim\|\vphi_0\|_{H^\sigma}^{\frac{s}{\sigma}}$$
	Analogous statements holds true if we replace $\vphi_t$ by its regularized version $\vphi_t^{(\alpha)}$ and $\nu$ by
	$$\nu^{(\alpha)}:=\sup_{|\tau|\leq T} \|\vphi_\tau^{(\alpha)}\|_{H^{\frac\gamma2}}.$$
\end{lemma}
\begin{proposition}[Well-posedness for low Sobolev regularity in the non-relativistic case]\label{frlowregwellposed}
	Let $\sigma=1$ and $s\in[0,1)$. Then \eqref{hartree} is globally well-posed in $H^s(\R^3)$, where we set $H^0(\R^3):=L^2(\R^3)$. More precisely, for any $0<T<\infty$ there is a unique solution $\vphi \in C\left([0,T); H^s(\R^3)\right) \cap  C^1\left([0,T);  H^{s-2}(\R^3)\right)$. Furthermore, the solution $\vphi$ continuously depends on the initial datum $\vphi_0$.
	\par Moreover, we have for any $t>0$
	\begin{equation*}
	\|\vphi_t\|_{H^s}\lesssim (t+1)\|\vphi_0\|_{H^s}.
	\end{equation*}
\end{proposition}

\begin{remark}
	Note that in the case $\sigma<\gamma/2$, we only obtain conditional well-posedness, as shown, e.g., in \cite{GZ}.
\end{remark}
The next two results are both generalizations and improvements of results given in \cite{MS}. For a proof, we refer to appendix \ref{hartreeregapprox}.
\begin{lemma}\label{fracregapprox}
	Suppose $\gamma\in(0,3/2)$ and $\sigma\in[\gamma/2,1]$. Let $\varepsilon\in[0,1]$, define $r:=\max\{\sigma;(1+\varepsilon)\gamma/2\}$, and $\vphi_0\in H^{r}(\R^3)$. Fix a time $T\in(0,\infty)$ such that
	$$\nu=\sup_{|\tau|\leq T} \|\vphi_\tau\|_{H^{\frac\gamma2}}<\infty.$$
	Let $\vphi$ be the solution of \eqref{hartree} and $\vphi^{(\alpha)}$ be a solution of \eqref{reghartree}. Then there is a constant $C=C_{\nu,\|\vphi_0\|_{H^{r}},T}$ such that for all $t\in[0,T]$ we have
	\begin{align*}
	\|\vphi_t-\vphi_t^{(\alpha)}\|_2&\leq C\alpha^{\frac{1+\varepsilon}2}\\
	\|(-\Delta)^{\frac\gamma4}(\vphi_t-\vphi_t^{(\alpha)})\|_2&\leq C\alpha^\varepsilon.			
	\end{align*}
\end{lemma}
As a direct consequence, we obtain the following.
\begin{corollary}\label{fracregregpers}
	Suppose $\gamma\in(0,3/2)$, $\sigma=\gamma/2$, and $\vphi_0\in H^{\gamma/2}(\R^3)$. Fix a time $T\in(0,\infty]$ such that
	$$\nu=\sup_{|\tau|\leq T} \|\vphi_\tau\|_{H^{\frac\gamma2}}<\infty.$$
	Let $\vphi$ be the solution of \eqref{hartree} and $\vphi^{(\alpha)}$ be a solution of \eqref{reghartree}. Then there is a constant $\kappa=\kappa_{\nu,\|\vphi_0\|_{H^{\gamma/2}},T}$ such that
	\begin{equation*}
	\sup_{t\in[0,T]}\sup_{\alpha\in(0,1)}\|\vphi_\tau^{(\alpha)}\|_{H^{\frac\gamma2}}\leq \kappa.
	\end{equation*}
\end{corollary}
\par Before continuing, let us introduce some notation first. Suppose $\Psi_{N,t}$ is a solution of \eqref{schroedinger}, $\vphi_t$ is a solution of \eqref{hartree}, $\vphi_t^{(\alpha)}$ is a solution of \eqref{reghartree}, and $\Psi_{N,t}^{(\alpha)}$ is given by \eqref{regschroedinger}.
We define the reduced density matrices 
\begin{equation*}
\begin{split}
\gamma_{N,t}^{(k)}\;&:=\;\tr_{k+1,...,N}|\Psi_{N,t} \rangle \langle \Psi_{N,t}|,\\
\gamma_{N,t}^{(k,\alpha)}\;&:=\;\tr_{k+1,...,N}\ket{\Psi_{N,t}^{(\alpha)}}\bra{\Psi_{N,t}^{(\alpha)}},
\end{split}
\end{equation*}
as well as the projections
\begin{equation*}
\begin{split}
P^{(k)}_t\;&:=\;\ket{\vphi_t^{\otimes k}} \bra{\vphi_t^{\otimes k}},\\
P^{(k,\alpha)}_t\;&:=\;\ket{(\vphi_t^{(\alpha)})^{\otimes k}} \bra{(\vphi_t^{(\alpha)})^{\otimes k}}.
\end{split}
\end{equation*}
Furthermore, we introduce the Pickl functionals, see \cite{P}, given by
\begin{equation}\label{def:aN}
\begin{split}
a_{N,t}\;&:=\;\scp{\Psi_{N,t}}{\left(1-(\ket{\vphi_t}\bra{\vphi_t})_1\right)\Psi_{N,t}},\\
a_{N,t}^{(\alpha)}\;&:=\;\scp{\Psi_{N,t}^{(\alpha)}}{\left(1-(\ket{\vphi_t^{(\alpha)}}\bra{\vphi_t^{(\alpha)}})_1\right)\Psi_{N,t}^{(\alpha)}}.
\end{split}
\end{equation}
Let for $r \in \R$ 
\begin{equation}\label{def:skr}
S_{k,r}:=\sum_{i=1}^k (1+S_i)^r
\end{equation}
and we denote the Hilbert-Schmidt norm of an operator acting on $L^2(\R^3)$ with $\|.\|_{HS}$. In the following, we will obtain convergence for 
$$\tr\left|S_{k,r}^{\frac12} (\gamma_{N,t}^{(k)}-P^{(k)}_t) S_{k,r}^{\frac12}\right|,$$
respectively for
$$\tr\left|S_{k,r}^{\frac12} (\gamma_{N,t}^{(k,\alpha)}-P^{(k)}_t) S_{k,r}^{\frac12}\right|.$$
Then the respective result in Theorem \ref{main} follows from \eqref{tracecyclicity}. Let us recall an important fact from \cite{AH}. 


\begin{proposition}[Anapolitanos, Hott \cite{AH}]\label{oldmainthm}
	For any $\theta \in [0,1)$ and any $s\geq0$ we have the estimate
	\begin{equation*}
	\tr\left|S_{k,\theta s}^{\frac12} (\gamma_{N,t}^{(k)}-P^{(k)}_t) S_{k,\theta s}^{\frac12}\right|\;\leq\; kC_{\Psi,\vphi,\theta,s} ( a_{N,t}^{\min\{\frac{1}{2};1-\theta\}} + \|\gamma_{N,t}^{(k)}-P^{(k)}_t\|_{HS}^{1-\theta}),
	\end{equation*}
	where $C_{\Psi_{N,t},\vphi_t,\theta,s}:=	2 \left(\|S_{1,s}^\frac{1}{2}\Psi_{N,t}\|_2 +\|S^\frac{s}{2}\vphi_t\|_2 \right)^{\max\{1;2 \theta\}}$. An analogous statement holds true if we replace $\Psi_{N,t}$ and $\vphi_t$ by their regularized analogues $\Psi_{N,t}^{(\alpha)}$ and $\vphi_t^{(\alpha)}$.
\end{proposition}
\begin{remark}\label{optimalregrem}
	Let us briefly discuss the best rates that we can expect with a result like Proposition \ref{oldmainthm}. The basic idea of this statement is to use boundedness in $H^1$-Sobolev trace norm and convergence in $L^2$-trace norm in order to interpolate the rate for Sobolev spaces in between. Assume an upper bound with $a_{N,t}^{1-\theta}$ rather than $a_{N,t}^{\min\{1/2;1-\theta\}}$ on the r.h.s. of the inequality in Proposition \ref{oldmainthm}, which is optimal. As mentioned above, the optimal rate on the $HS$-norm is $1/N$ as for the Pickl functional $a_{N,t}$, see below. These results would give us then the optimal rate of $1/N^{1-\theta}$ with the present techniques; the heuristics is based on the $H^1$-norm being held stationary, while there is convergence with rate $1/N$ in $L^2$. (In)formal interpolation would then give us $1/N^{1-\theta}$.
\end{remark}
Combining this result with \eqref{tracecyclicity}, we reduce showing convergence as stated in \eqref{convtype} to showing convergence for both $a_{N,t}$ and $\|\gamma_{N,t}^{(k)}-P_t^{(k)}\|_{HS}$. This will be the goal of the subsequent.
\par The following theorem has not been explicitly proven before, but can be easily obtained by following the steps of \cite{P} or \cite{KP}.
\begin{proposition}\label{firstprevious}
	Assume either of \ref{itm:sr}, \ref{itm:fs}, \ref{itm:nr}. Suppose $N\in\N$ and $\vphi_0\in H^\gamma(\R^3)$. Let $\Psi_{N,t}$ be a solution of \eqref{schroedinger} and $\vphi_t$ be a solution of \eqref{hartree}. Then we have for all $t>0$
	\begin{equation*}
	a_{N,t}\;\leq\;\frac{\e^{c\int_0^t\|\frac1{|\cdot|^{2\gamma}}*|\vphi_\tau|^2\|_\infty^{\frac12}\mathrm{d}\tau}}{N},
	\end{equation*}
	where $c$ is universal.
\end{proposition}
\begin{remark}
	Note that the weak Young's inequality together with the Sobolev embedding implies that the integrand in the exponent is bounded by $\|\vphi_{\tau}\|_{H^\gamma}$. Lemma \ref{fracpersreg} then yields at most a super-exponential growth of $t\mapsto a_{N,t}$.
\end{remark}
\par As, e.g., presented in \cite{AHH} and \cite{AH}, we have
\begin{equation}\label{picklschmidt}
\begin{split}
\tr{\left|\gamma_{N,t}^{(k)}-P^{(k)}_t\right|}&\lesssim\|\gamma_{N,t}^{(k)}-P^{(k)}_t\|_{HS}\lesssim\sqrt{k a_{N,t}},\\
\tr{\left|\gamma_{N,t}^{(k,\alpha)}-P^{(k,\alpha)}_t\right|}&\lesssim \|\gamma_{N,t}^{(k,\alpha)}-P^{(k,\alpha)}_t\|_{HS}\lesssim\sqrt{k a_{N,t}^{(\alpha)}}
\end{split}
\end{equation}
The respective first inequality goes back to an argument by R. Seiringer, the second one was brought to the author's attention by M. Griesemer.
\par By Proposition \ref{firstprevious}, it is sufficient to bound
$$\int_0^t\|\frac1{|\cdot|^{2\gamma}}*|\vphi_\tau|^2\|_\infty^{\frac12}\mathrm{d}\tau,$$
in order to show convergence of the Pickl functional and, by the above comments and Proposition \ref{oldmainthm}, convergence in Sobolev trace norms. For that, we want to use an observation made in \cite{CLL} and improve it: Instead of writing 
\begin{equation*}
\frac1{|\cdot|^\gamma}\;=\;V_2+V_\infty	
\end{equation*}
with $V_2\in L^2(\R^3)$ and $V_\infty\in L^\infty(\R^3)$, we want to employ the fact that $|\cdot|^{-\gamma}\in L^{3/\gamma}_w(\R^3)$. This shall give us better results in the subsequent, as we will see below. By Young's inequality and using mass conservation that, we find that
$$\|\frac1{|\cdot|^{2\gamma}}*|\vphi_\tau|^2\|_\infty\lesssim\|\vphi_{\tau}\|_{6/(3-2\gamma)}^2.$$
In particular, the above mentioned integral can be bounded via
\begin{equation}\label{potentialstrichartz}
\int_0^t\|\frac1{|\cdot|^{2\gamma}}*|\vphi_\tau|^2\|_\infty^{\frac12}\mathrm{d}\tau\lesssim\int_0^t\|\vphi_\tau\|_{6/(3-2\gamma)}\mathrm{d}\tau=\|\vphi_{\tau}\|_{L^1_\tau([0,t];L^{6/(3-2\gamma)}_x)}.
\end{equation}
Whereas \cite{CLL} use the above decomposition of the interaction potential together with Strichartz estimates to control more singular potentials, we want to employ the fact that $|\cdot|^{-\gamma}\in L^{3/\gamma}_w(\R^3)$ together with Strichartz estimates both for considering more singular potentials and lowering the required regularity on the initial datum $\vphi_0$. 
\par Altogether, the program we will run goes as follows:
\begin{enumerate}
	\item Reduce needed regularity of initial data by means of Strichartz estimates when getting control over the Pickl functional $a_{N,t}$.
	\item Control the Hilbert-Schmidt norm by $a_{N,t}$.
	\item Control higher Sobolev trace norms.
\end{enumerate}

\begin{remark}
	With the present method, the optimality of the regularity of initial data fully relies	optimality of the respectively used Strichartz estimates. Rather than providing optimal regularity, this work aims to provide a guide how to derive the Hartree equation and improving the needed regularity. Also, we consider the general case of not necessarily radial solutions of \eqref{hartree}. In the case of radial solutions, there are improved Strichartz estimates, see \cite{CKS}, \cite{CL}, \cite{GW}.
\end{remark}

	\section{Derivation of the Hartree equation}
	\subsection{The semi-relativistic case}
	In this section, always assume \ref{itm:sr}. Our first goal is to reduce the needed regularity for showing convergence towards the Hartree equation in trace norm. As mentioned above, we will apply Strichartz estimates in order to decrease the required regularity on $\vphi$, which by Lemma \ref{fracpersreg} reduces to required regularity on $\vphi_0$. Let us recall Strichartz estimate for the semi-relativistic NLS, also known as the half-wave equation, from \cite{D}. We call a pair $(q,r)\in[2,\infty]^2$ \textit{admissible} iff $(q,r)\neq(2,\infty)$ and 
	$$\frac1q+\frac1r=\frac12.$$
	In addition, define a residual power
	$$\alpha_{q,r}:=\frac32-\frac3r-\frac1q.$$
	\begin{lemma}[Dinh \cite{D}]\label{srstrichartz}
		Let $(q,r)\in[2,\infty]^2$ be admissible, $\alpha\in\R$. Then we have for any interval $I\subseteq\R$
			\begin{equation*}
			\begin{split}
			\|\e^{-\i(-\Delta)^{1/2}t}\vphi_0\|_{L^q_t(I;W^{\alpha,r}_x)}\;&\lesssim\;\|\vphi_0\|_{H^{\alpha+\alpha_{q,r}}},\\
			\|\int_0^t \e^{-\i (-\Delta)^{1/2}(t-\tau)}F(\tau) \mathrm{d}\tau\|_{L^q_t(I;W^{\alpha,r}_x)}\;&\lesssim\;\|F\|_{L^1_t(I;H^{\alpha+\alpha_{q,r}}_x)}.
			\end{split}
			\end{equation*}
	\end{lemma}
\paragraph{Defocusing case or focusing case with small coupling.} Let us consider the regime $\mu=1$ or $\lambda<\lambda_{S,c}$.
	\begin{proposition}\label{srsupconv}
		Assume $\vphi_0\in H^{2/3}(\R^3)$, $N\in\N$, and $k\in\N^{\leq N}$. Let $\Psi_{N,t}$ be a solution of \eqref{schroedinger} and $\vphi_t$ be a solution of \eqref{hartree}. Then there is a constant $C=C_{\|\vphi_0\|_{H^{2/3}}}$ such that for all $t\geq0$ we have
		\begin{equation}
		\begin{split}
			\tr{\left|\gamma_{N,t}^{(k)}-P^{(k)}_t\right|}\;&\lesssim\;\sqrt{k}\frac{\e^{C\e^{Ct}}}{\sqrt{N}},\\
			a_{N,t}\;&\leq\;\frac{\e^{C\e^{Ct}}}{N}.
		\end{split}			
		\end{equation}
	\end{proposition}
\begin{proof}
	By applying first H\"older's inequality followed by the Strichartz estimate \ref{srstrichartz}, we obtain
	\begin{align*}
		\|\vphi\|_{L^1_\tau([0,t];L^{6}_x)}&\leq t^{\frac23}\|\vphi\|_{L^3_\tau([0,t];L^{6}_x)}\lesssim t^{\frac23}\left(\|\vphi_0\|_{H^{\frac23}}+\|(\frac1{|\cdot|}*|\vphi_{\tau}|^2)\vphi_{\tau}\|_{L^1_\tau([0,t];H^{\frac23}_x)}\right)\\
		&\lesssim t^{\frac23}\left(\|\vphi_0\|_{H^{\frac23}}+\int_0^t\|\vphi_\tau\|_{H^{\frac12}}^2\|\vphi_\tau\|_{H^{\frac23}}\mathrm{d}\tau \right),
	\end{align*}
	where in the last step we applied Lemma \ref{fraccontract} in the appendix. Lemma \ref{fracpersreg} together with energy conservation yields
	$$\|\vphi\|_{L^1_\tau([0,t];L^{6}_x)}\lesssim C\e^{Ct},$$
	for some constant $C=C_{\|\vphi_0\|_{H^{2/3}}}$. Together with \eqref{potentialstrichartz} and \eqref{picklschmidt} this implies the statement.
\end{proof}

Collecting the previous results, we have proved the following theorems.
\begin{proposition}\label{srsubcrit}
	Assume $\vphi_0\in H^{2/3}(\R^3)$. Suppose $N\in\N$ and $k\in\N^{\leq N}$. Let $\Psi_{N,t}$ be a solution of \eqref{schroedinger} and $\vphi_t$ be a solution of \eqref{hartree}. Then there is a constant $C=C_{\|\vphi_0\|_{H^{2/3}}}$ such that for any $\theta\in[0,1)$ and any $t\geq0$ we have
	$$\tr\left|S_{k,\theta}^{\frac12} (\gamma_N^{(k)}-P^{(k)}) S_{k,\theta}^{\frac12}\right|\leq Ck^{\frac{3-\theta}2}\frac{\e^{C\e^{Ct}}}{N^{\frac{1-\theta}{2}}}.$$
\end{proposition}
\begin{proof}
	In view of Proposition \ref{oldmainthm}, Proposition \ref{srsupconv} and \eqref{picklschmidt}, it only remains to show uniform boundedness of $\|S_{1,1}^{1/2}\Psi_{N,t}\|_2 +\|S^{1/2}\vphi_t\|_2$. Energy conservation directly implies
	$$\|S^{1/2}\vphi_t\|_2\lesssim \|\vphi_0\|_{H^{1/2}}.$$
	Moreover, we have, using \eqref{hamsupercrit} and energy conservation
	\begin{equation}
		\begin{split}
			\|S_{1,1}^{1/2}\Psi_{N,t}\|_2^2&=\frac1N\scp{\Psi_{N,t}}{\left(\sum_{i=1}^NS_i+N\right)\Psi_{N,t}}\\
			&\lesssim\frac1N\scp{\Psi_{N,t}}{(H_N+N)\Psi_{N,t}}\\
			&=\frac1N\scp{\Psi_{N,0}}{(H_N+N)\Psi_{N,0}}\\
			&\lesssim\|\vphi_0\|_{H^{1/2}}^2,	
		\end{split}
	\end{equation}
	where in the last step we used that the total energy can be bounded by the kinetic energy in the present case. 
\end{proof}
Let us recall the result
$$\|\gamma_{N,t}^{(1)}-P^{(1)}_t\|_{HS}\leq\frac{f(t)}{N}$$
given in \cite{Lee}. In here $f(t)$ depends only on $\lambda$ and $\sup_{\tau\in[0,t]}\|\vphi_{\tau}\|_{H^1}$. Note that due to Lemma \ref{fracpersreg}, we have
$$\sup_{\tau\in[0,t]}\|\vphi_{\tau}\|_{H^1}\lesssim \e^{c\nu^2t}\|\vphi_0\|_{H^1},$$
where $c$ is a universal constant and $\nu=\sup_{\tau\in[0,t]}\|\vphi_{\tau}\|_{H^{1/2}}$ which, by energy conservation, can be uniformly bounded by a multiple of $\|\vphi_0\|_{H^{1/2}}.$ Thus we have that the above function $C(t)$ is actually a function only of $\lambda$, $\|\vphi_0\|_{H^1}$ and $t$.
\par  Combining this remark with Proposition \ref{oldmainthm}, Proposition \ref{srsupconv}, Lemma \ref{fracpersreg}, and the proof of Proposition \ref{srsubcrit}, we obtain the following result.
\begin{proposition}\label{srlee}
	Assume $\vphi_0\in H^{1}(\R^3)$ and $N\in\N$. Let $\Psi_{N,t}$ be a solution of \eqref{schroedinger} and $\vphi_t$ be a solution of \eqref{hartree}. Then there is a function $f(t)=f(t,\|\vphi_0\|_{H^{1}})$ such that for any $\theta\in[0,1)$ and any $t\geq0$ we have
	$$\tr\left|S_{1,\theta}^{\frac12} (\gamma_{N,t}^{(1)}-P^{(1)}_t) S_{1,\theta}^{\frac12}\right|\leq \frac{f(t)}{N^{\min\{1/2;1-\theta\}}}.$$
\end{proposition}

	\paragraph{The focusing case.}  In this paragraph, assume $\mu=-1$ and $\lambda\geq\lambda_{S,c}$. It is well-known that solutions of the Hartree equation \eqref{hartree} exhibit blow-up after finite time, see, e.g. \cite{Le}. Even worse, as mentioned above, $H_N$ is not necessarily a self-adjoint operator for which we could solve \eqref{schroedinger}. Thus, we cannot get a control over $a_{N,t}$ as defined in \eqref{def:aN}. Instead, we work with a regularized Pickl functional, see \eqref{def:aN},
	\begin{equation*}
		a_{N,t}^{(\alpha)}\;=\;\scp{\Psi_{N,t}^{(\alpha)}}{\left(1-(\ket{\vphi_t^{(\alpha)}}\bra{\vphi_t^{(\alpha)}})_1\right)\Psi_{N,t}^{(\alpha)}}.
	\end{equation*}
	We are able to state our main theorem for this section.
	\begin{proposition}\label{srpicklsupercrit}
		Let $\vphi_0\in H^{2/3}(\R^3)$. Fix some $T\in(0,\infty)$ such that
		$$\nu\;=\;\sup_{|\tau|\leq T} \|\vphi_\tau\|_{H^{\frac12}}<\infty.$$
		Suppose $N\in\N$ and $k\in\N^{\leq N}$. Let $\Psi_{N,t}^{(\alpha)}=\e^{-\i H_N^{(\alpha)}t}\vphi_0^{\otimes N}$ and $\vphi$ be a solution of \eqref{hartree}. Then there is a constant $C=C_{\nu,\|\vphi_0\|_{H^{2/3}},T}$ such that with the above notations we have for any $0<\alpha<1$ and any $t\in [0,T]$
		\begin{equation*}
		\begin{split}
			\tr{\left|\gamma_{N,t}^{(k,\alpha)}-P^{(k)}_t\right|}&\leq C\sqrt{k}\left(\frac{1}{\sqrt{N}}+\alpha^{\frac23}\right),\\
			a_{N,t}^{(\alpha)}&\leq\frac{C}{N}.
		\end{split}
		\end{equation*}
		In particular, we have for $\alpha=\alpha_N=\mathcal{O}(N^{-3/4})$
		\begin{equation*}
		\tr{\left|\gamma_{N,t}^{(k,\alpha_N)}-P^{(k)}_t\right|}\lesssim C \sqrt{\frac{k}{N}}.
		\end{equation*}
		
	\end{proposition}
	\begin{proof}[Proof of Proposition \ref{srpicklsupercrit}]
		As mentioned, e.g., in \cite{AH}, using the variational characterization of the first eigenvalue, one can show with the above notation
		\begin{equation}\label{srtracenormsplit}
		\begin{split}
		\tr{\left|\gamma_{N,t}^{(k,\alpha)}-P^{(k)}_t\right|}&\lesssim\|\gamma_{N,t}^{(k,\alpha)}-P^{(k)}_t\|_{HS}\\
		&\lesssim \|\gamma_{N,t}^{(k,\alpha)}-P^{(k,\alpha)}_t\|_{HS}+\|P^{(k,\alpha)}_t-P^{(k)}_t\|_{HS}\\
		&	\lesssim\sqrt{ka_{N,t}^{(\alpha)}}+\sqrt{k}\|\vphi_t^{(\alpha)}-\vphi_t\|_2,
		\end{split}
		\end{equation}
		where in the last we used \eqref{picklschmidt} together with the fact that the $HS$-distance of two rank-1 projections is bounded from above by the respective $L^2$-distance of the states onto which we project. By Lemma \ref{fracregapprox}, there is a constant $C=C_{\nu,\|\vphi_0\|_{H^{2/3}},T}$ such that for all $t\in[0,T]$ we have
		\begin{equation}\label{srregpickll2dist}
		\|\vphi_t-\vphi_t^{(\alpha)}\|_2\leq C\alpha^{\frac23}.	
		\end{equation}
		Next, following the steps of \cite{P} and \cite{AHH}, we can equally show
		\begin{equation}\label{sranregbd}
		a_{N,t}^{(\alpha)}\;\leq\;\frac{\e^{\int_0^t\|\frac1{|\cdot|^{2}}*|\vphi_\tau^{(\alpha)}|^2\|_\infty^{\frac12}\mathrm{d}\tau}}{N}.
		\end{equation}
		With analogous steps as in the proof of Proposition \ref{srsupconv}, one can show
		\begin{equation*}
			\int_0^t\|\frac1{|\cdot|^{2}}*|\vphi_\tau^{(\alpha)}|^2\|_\infty^{\frac12}\mathrm{d}\tau \lesssim t^{\frac23}\left(\|\vphi_0\|_{H^{\frac23}}+\int_0^t\|\vphi_\tau^{(\alpha)}\|_{H^{\frac12}}^2\|\vphi_\tau^{(\alpha)}\|_{H^{\frac23}}\mathrm{d}\tau \right).
		\end{equation*}
		Lemma \ref{fracpersreg} together with Corollary \ref{fracregregpers} and energy conservation then yields
		$$\|\vphi_\tau^{(\alpha)}\|_{H^{\frac12}}^2\|\vphi_\tau^{(\alpha)}\|_{H^{\frac23}}\leq C$$
		for some $C=C_{\nu,\|\vphi_0\|_{H^{2/3}},T}$. This concludes the proof.
	\end{proof}
	With this result at hand, we can even show convergence in higher Sobolev trace norms.
	\begin{proposition}\label{srhighersobolevconv}
		Let $r\geq2/3$, $\vphi_0\in H^{r}(\R^3)$ and define $\varepsilon\in[1/3,1]$ to satisfy $(1+\varepsilon)/2\geq r$. Fix some $T\in(0,\infty)$ such that
		$$\nu\;=\;\sup_{|\tau|\leq T} \|\vphi_\tau\|_{H^{\frac12}}<\infty.$$
		Suppose $N\in\N$ and $k\in\N^{\leq N}$. Let $\Psi_{N,t}^{(\alpha)}=\e^{-\i H_N^{(\alpha)}t}\vphi_0^{\otimes N}$ and $\vphi$ be a solution of \eqref{hartree}. Then there is a constant $C=C_{\nu,\|\vphi_0\|_{H^{r}},T}$ such that with the above notations we have for any $0<\alpha<1$ small enough, any $t\in [0,T]$ and any $\theta\in[0,\min\{\frac{\vep}{1-\vep};1\})$
		$$\tr\left|S_{k,\theta}^{\frac12} (\gamma_{N,t}^{(k,\alpha)}-P^{(k)}_t) S_{k,\theta}^{\frac12}\right|\leq C k^{\frac{3-\theta}2}\left( \frac1{\alpha^{\max\{1/2;\theta\}}N^{\frac{1-\theta}{2}}}+\alpha^{\frac{1-\theta+\varepsilon(1+\theta)-\max\{1;2\theta\}}{2}}\right).$$
		In particular, we have for $\alpha=\alpha_N=\mathcal{O}(N^{-(1-\theta)/(1-\theta+\varepsilon(1+\theta))})$
		$$\tr\left|S_{k,\theta}^{\frac12} (\gamma_{N,t}^{(k,\alpha)}-P^{(k)}_t) S_{k,\theta}^{\frac12}\right|\leq C\frac{k^{\frac{3-\theta}2}}{N^{\frac{(1-\theta)(1-\theta+\varepsilon(1+\theta)-\max\{1;2\theta\})}{2[1-\theta+\varepsilon(1+\theta)]}}}.$$
	\end{proposition}
	\begin{proof}
		We start with the estimate
		\begin{equation}\label{srsobolevsplit}
		\begin{split}
		\tr\left|S_{k,\theta}^{\frac12} (\gamma_{N,t}^{(k,\alpha)}-P^{(k)}_t) S_{k,\theta}^{\frac12}\right|\lesssim  	\tr\left|S_{k,\theta}^{\frac12} (\gamma_{N,t}^{(k,\alpha)}-P^{(k,\alpha)}_t) S_{k,\theta}^{\frac12}\right|+\tr\left|S_{k,\theta}^{\frac12} (P^{(k,\alpha)}_t-P^{(k)}_t) S_{k,\theta}^{\frac12}\right|.
		\end{split}
		\end{equation}
		 For the first term, we apply Proposition \ref{oldmainthm} to get
		 \begin{equation}\label{srtracenormpart}
		 	\begin{split}
			 	\tr\left|S_{k,\theta}^{\frac12} (\gamma_N^{(k,\alpha)}-P^{(k,\alpha)}) S_{k,\theta}^{\frac12}\right|\lesssim kC_{\Psi_{N,t}^{(\alpha)},\vphi_t^{(\alpha)},\theta} \left( (a_{N,t}^{(\alpha)})^{\min\{1/2;1-\theta\}} + \|\gamma_{N,t}^{(k,\alpha)}-P^{(k,\alpha)}_t\|_{HS}^{1-\theta}\right),	 	
		 	\end{split}
		 \end{equation}
		 where $C_{\Psi_{N,t}^{(\alpha)},\vphi_t^{(\alpha)},\theta}:=\left(\|S_{1,1}^\frac{1}{2}\Psi_{N,t}^{(\alpha)}\|_2 +\|S^\frac{1}{2}\vphi_t^{(\alpha)}\|_2 \right)^{\max\{1;2 \theta\}}$. By Corollary \ref{fracregregpers}, we have $\|S^\frac{1}{2}\vphi_t^{(\alpha)}\|_2\leq\kappa$ for some $\kappa=\kappa_{\nu,\|\vphi_0\|_{H^{1/2}},T}$. In addition, using symmetry of $\Psi_{N,t}$ w.r.t. particle permutations followed by \eqref{reghamboundskin} and energy conservation for the Schr\"odinger equation, we find
		 \begin{equation}\label{srregnorm}
		 	\begin{split}
				\|S_{1,1}^\frac{1}{2}\Psi_{N,t}^{(\alpha)}\|_2^2-1&=\frac1N\scp{\Psi_{N,t}^{(\alpha)}}{\sum_{i=1}^NS_i\Psi_{N,t}^{(\alpha)}}\\
				&\leq\frac1N\scp{\Psi_{N,t}^{(\alpha)}}{H_N^{(\alpha)}\Psi_{N,t}^{(\alpha)}}+\frac{\lambda}{2\alpha}\\
				&=\frac1N\scp{\Psi_{N,0}^{(\alpha)}}{H_N^{(\alpha)}\Psi_{N,0}^{(\alpha)}}+\frac{\lambda}{2\alpha}\\
				&\leq\|\vphi_0\|_{H^{1/2}}^2+\frac{\lambda}{2\alpha}.				
		 	\end{split}
		 \end{equation}
		Using the bound on $a_{N,t}^{(\alpha)}$ in Proposition \ref{srpicklsupercrit} together with \eqref{picklschmidt} gives
		 $$\|\gamma_{N,t}^{(k,\alpha)}-P^{(k,\alpha)}_t\|_{HS}\leq C\sqrt{\frac{k}{N}}$$
		 for some $C=C_{\nu,\|\vphi_0\|_{H^{2/3}},T}$. If we then choose $\alpha$ small enough, we thus obtain using \eqref{srtracenormpart} together with \eqref{srregnorm} and Proposition \ref{srpicklsupercrit}
		 \begin{equation}
		 		 \tr\left|S_{k,\theta}^{\frac12} (\gamma_N^{(k,\alpha)}-P^{(k,\alpha)}) S_{k,\theta}^{\frac12}\right|\lesssim C\frac{k^{\frac{3-\theta}2}}{\alpha^{\max\{1/2;\theta\}}N^{\frac{1-\theta}2}}
		 \end{equation}
		 for some possibly bigger constant $C=C_{\nu,\|\vphi_0\|_{H^{2/3}},T}$.
		 \par For the second term in \eqref{srsobolevsplit}, we use the variational characterization of the first eigenvalue to obtain
		 \begin{equation}\label{srprojectionssplit}
		 	\tr\left|S_{k,\theta}^{\frac12} (P^{(k,\alpha)}_t-P^{(k)}_t) S_{k,\theta}^{\frac12}\right|\lesssim \tr\left(S_{k,\theta}^{\frac12} (P^{(k,\alpha)}_t-P^{(k)}_t) S_{k,\theta}^{\frac12}\right)+\|S_{k,\theta}^{\frac12} (P^{(k,\alpha)}_t-P^{(k)}_t) S_{k,\theta}^{\frac12}\|_{HS},
		 \end{equation}
		 a fact first pointed out by R. Seiringer, see also \cite{AH}. To bound the first term, we use
		 \begin{equation}\label{srprojetraceaid}
		 \begin{split}
			\tr\left(S_{k,\theta}^{\frac12} (P^{(k,\alpha)}_t-P^{(k)}_t) S_{k,\theta}^{\frac12}\right)&=\|S_{k,\theta}^{\frac12}(\vphi_t^{(\alpha)})^{\otimes k}\|_2^2-\|S_{k,\theta}^{\frac12}\vphi_t^{\otimes k}\|_2^2\\
			&=k\left(\|S^\frac{\theta}{2}\vphi_t^{(\alpha)}\|_2^2-\|S^\frac{\theta}{2}\vphi_t\|_2^2\right)\\
			&\leq k\scp{|S^\frac{\theta}{2}(\vphi_t^{(\alpha)}-\vphi_t)|}{|S^\frac{\theta}{2}(\vphi_t^{(\alpha)}+\vphi_t)|}\\
			&\leq k\|S^\frac{\theta}{2}(\vphi_t^{(\alpha)}-\vphi_t)\|_2(\|S^\frac{\theta}{2}\vphi_t^{(\alpha)}\|_2+\|S^\frac{\theta}{2}\vphi_t\|_2),
		 \end{split}	 	
		 \end{equation}
		 where in the last step we applied Cauchy-Schwarz. Using interpolation together with Lemma \ref{fracregapprox}, with $\varepsilon=1/3$, and Corollary \ref{fracregregpers}, we then obtain
		 \begin{equation}\label{srprojetrace}
			 \begin{split}
				\tr\left(S_{k,\theta}^{\frac12} (P^{(k,\alpha)}_t-P^{(k)}_t) S_{k,\theta}^{\frac12}\right)&\leq Ck\|\vphi_t^{(\alpha)}-\vphi_t\|_2^{1-\theta}\|S^{1/2}(\vphi_t^{(\alpha)}-\vphi_t)\|_2^{\theta}\\
				&\leq Ck\alpha^{\frac{1-\theta+\varepsilon(1+\theta)}{2}},
			 \end{split}
		 \end{equation}
		 where the involved constants depend on $C=C_{\nu,\|\vphi_0\|_{H^{2/3}},T}$. 
		 \par To bound the second term in \eqref{srprojectionssplit}, we use that we can write the Hilbert-Schmidt norm of on operator as the $L^2$-norm of its kernel. Abbreviating $f:=S_{k,\theta}^{1/2}(\vphi_t)^{\otimes k}$ and $g:=S_{k,\theta}^{1/2}(\vphi_t^{(\alpha)})^{\otimes k}$, we thus have
		 \begin{equation}\label{srhs}
		 \begin{split}
		 	\|S_{k,\theta}^{\frac12} (P^{(k,\alpha)}_t-P^{(k)}_t) S_{k,\theta}^{\frac12}\|_{HS}&=\|\bar{f}\otimes f-\bar{g}\otimes g\|_2\\
		 	&\leq\|(\bar{f}-\bar{g})\otimes f\|_2+\|\bar{g}\otimes(f-g)\|_2\\
		 	&=\|f-g\|_2(\|f\|_2+\|g\|_2).
		 \end{split}
		 \end{equation}
		 Next, we compute
		 \begin{equation}\label{srfgdiff}
		 	\begin{split}
				\|f-g\|_2^2&=k\|S^\frac{\theta}{2}(\vphi_t^{(\alpha)}-\vphi_t)\|_2^2\\
				&\leq k\|\vphi_t^{(\alpha)}-\vphi_t\|_2^{2-2\theta}\|S^{1/2}(\vphi_t^{(\alpha)}-\vphi_t)\|_2^{2\theta}\\
				&\leq Ck\alpha^{1-\theta+\varepsilon(1+\theta)}
		 	\end{split}
		 \end{equation}
		 for some constant $C=C_{\nu,\|\vphi_0\|_{H^{2/3}},T}$, where we used the same arguments as for \eqref{srprojetrace}. Similarly, we show
		 \begin{equation}\label{srfgsum}
		 	\|f\|_2+\|g\|_2\leq C\sqrt{k}
		 \end{equation}
		 for some constant $C=C_{\nu,\|\vphi_0\|_{H^{2/3}},T}$. Collecting \eqref{srprojetrace}, \eqref{srhs}, \eqref{srfgdiff}, and \eqref{srfgsum}, we obtain the desired result.
	\end{proof}
	As an easy corollary of Proposition \ref{srpicklsupercrit} and the result given in \cite{Lee}, we also obtain the following special case.
	\begin{proposition}
		Let $\vphi_0\in H^1(\R^3)$. Fix some $T\in(0,\infty)$ such that
		$$\nu\;=\;\sup_{|\tau|\leq T} \|\vphi_\tau\|_{H^{\frac12}}<\infty.$$
		Suppose $N\in\N$ and $k\in\N^{\leq N}$. Let $\Psi_{N,t}^{(\alpha)}=\e^{-\i H_N^{(\alpha)}t}\vphi_0^{\otimes N}$ and $\vphi$ be a solution of \eqref{hartree}. Then there is a constant $C=C_{\nu,\|\vphi_0\|_{H^1},T}$ such that with the above notations we have for any $0<\alpha<1$ small enough, any $t\in [0,T]$ and any $\theta\in[0,1)$
		$$\tr\left|S_{1,\theta}^{\frac12} (\gamma_{N,t}^{(1,\alpha)}-P^{(1)}_t) S_{1,\theta}^{\frac12}\right|\leq C \left( \frac1{\alpha^{\max\{1/2;\theta\}}N^{\min\{1/2;1-\theta\}}}+\alpha^{1-\max\{1/2;\theta\}}\right).$$
		In particular, we have for $\alpha=\alpha_N=\mathcal{O}(N^{-\min\{1/2;1-\theta\}})$
		$$\tr\left|S_{1,\theta}^{\frac12} (\gamma_{N,t}^{(1,\alpha)}-P^{(1)}_t) S_{1,\theta}^{\frac12}\right|\leq \frac{C}{N^{\min\{1/2;1-\theta\}^2}}.$$
	\end{proposition}
\begin{proof}
	As in the last proof, we start with the estimate
	\begin{equation*}
		\begin{split}
		\tr\left|S_{1,\theta}^{\frac12} (\gamma_{N,t}^{(1,\alpha)}-P^{(1)}_t) S_{1,\theta}^{\frac12}\right|
		&\lesssim C_{\Psi_{N,t}^{(\alpha)},\vphi_t^{(\alpha)},\theta} \left( (a_{N,t}^{(\alpha)})^{\min(\frac{1}{2},1-\theta)} + \|\gamma_{N,t}^{(1,\alpha)}-P^{(1,\alpha)}_t\|_{HS}^{1-\theta}\right)\\
		&\qquad +\tr\left|S_{1,\theta}^{\frac12} (P^{(1,\alpha)}_t-P^{(1)}_t) S_{1,\theta}^{\frac12}\right|.
		\end{split}
	\end{equation*}
	Whereas for the remaining terms, we use the same bounds as established in the previous proof, we want to use the better rate of \cite{Lee} to bound $\|\gamma_{N,t}^{(1,\alpha)}-P^{(1,\alpha)}_t\|_{HS}\leq C/N$. We have
	\begin{align*}
			\|\gamma_{N,t}^{(1,\alpha)}-P^{(1,\alpha)}_t\|_{HS}&\leq\|\gamma_{N,t}^{(1,\alpha)}-P^{(1)}_t\|_{HS}+\|P^{(1)}_t-P^{(1,\alpha)}_t\|_{HS}\\
			&\leq\|\gamma_{N,t}^{(1,\alpha)}-P^{(1)}_t\|_{HS}+\|\vphi_t^{(\alpha)}-\vphi_t\|_2\\
			&\leq C\left(\frac1N+\alpha\right)
	\end{align*}
	for some constant $C=C_{\nu,H^1,T}$, where in the last step we used the above mentioned result of \cite{Lee} together with Lemma \ref{fracregapprox}, with $\varepsilon=1$. Repeating the same steps of the last proof, together with Lemma \ref{fracregapprox}, with $\varepsilon=1$, we obtain
	$$\tr\left|S_{1,\theta}^{\frac12} (P^{(1,\alpha)}_t-P^{(1)}_t) S_{1,\theta}^{\frac12}\right|\leq C\alpha$$
	for some constant $C=C_{\nu,H^1,T}$. This concludes the proof.
\end{proof}

	\begin{remark}[Some comments on the proofs of Propositions \ref{srsupconv} -- \ref{srhighersobolevconv}]\label{simplifyrest}
		We want to point out that the only arguments that involved the specific parameters $(\gamma,\sigma)$ of the Hamiltonian were used in bounding $\|\vphi_t^{(\alpha)}-\vphi_t\|_2$, $\|S^{1/2}(\vphi_t^{(\alpha)}-\vphi_t)\|_2$ as well, when applying Strichartz estimates to obtain an upper bound on $a_{N,t}$ and $a_{N,t}^{(\alpha)}$. Moreover, that we use the same estimates for both $a_{N,t}$ and $a_{N,t}^{(\alpha)}$. More precisely, we show  
		\begin{equation*}
			\begin{split}
				a_{N,t}&\leq\frac{\e^{t^{p}\left(\|\vphi_0\|_{H^{s}}+\int_0^t\|\vphi_\tau\|_{H^{s'}}^3\mathrm{d}\tau \right)+t}}{N},\\
				a_{N,t}^{(\alpha)}&\leq\frac{\e^{t^{p}\left(\|\vphi_0\|_{H^{s}}+\int_0^t\|\vphi_\tau^{(\alpha)}\|_{H^{s'}}^3\mathrm{d}\tau \right)+t}}{N}
			\end{split}
		\end{equation*}
		for some exponents $p$, $s$, $s'$. By Lemma \ref{fracpersreg} this reduces to only have a uniform bound on $\|\vphi_t\|_{H^{\gamma/2}}$ respectively $\|\vphi_t^{(\alpha)}\|_{H^{\gamma/2}}(\R^3)$. Corollary \ref{fracregregpers} then shows that this reduces to requiring a uniform bound on $\|\vphi_t\|_{H^{\gamma/2}}$. We will use this observation to shorten the proofs in the remaining sections.
	\end{remark}

	\subsection{Fractional NLS with possibly singular potentials}
	
	In this section, we always assume \ref{itm:fs}. For this is rather non-physical case, not much interest has been shown in deriving the NLS from the Schr\"odinger equation. Since our present tools are sufficient and sufficiently simple, we want to present a derivation here. In order to reduce needed regularity when applying the Pickl method, we again need Strichartz estimates. A pair $(q,r)\in[2,\infty]^2$ is called \textit{admissible} iff
	$$\frac2q+\frac3r=\frac32.$$
	Define the following Strichartz-norm
	\begin{equation}\label{strichartznorm}
		\|u\|_{S^s_{q,r}(I)}\;:=\;\||\nabla|^{-3(1-\sigma)(\frac12-\frac1r)}u_\tau\|_{L^q_t(I;W^{s,r}_x)}
	\end{equation}
	for an interval $I\subseteq \R$. Let us recall the following fact from \cite{COX}, see also \cite{HS}.
	\begin{lemma}[Cho, Ozawa, Xia \cite{COX}]\label{fsstrichartz}
		Let $(q,r)\in[2,\infty]^2$ be admissible. Then we have for any interval $I\subseteq\R$ and any $s\in\R$
		\begin{equation*}
		\begin{split}
		\|\e^{-\i (-\Delta)^{\sigma}t}\vphi_0\|_{S^s_{q,r}(I)}\;&\lesssim\;\|\vphi_0\|_{H^s},\\
		\|\int_0^t \e^{-\i (-\Delta)^{\sigma}(t-\tau)}F(\tau) \mathrm{d}\tau\|_{S^s_{q,r}(I)}\;&\lesssim\;\|F\|_{L^1_t([0,t];H^s_x)}.
		\end{split}
		\end{equation*}
	\end{lemma}
	\paragraph{Defocusing case or focusing case with small coupling.} Let us consider the regime $\mu=1$ or $\lambda<\lambda_{S,c}$.
\begin{proposition}\label{fssupconv}
Let $r:=\max\{\sigma;(1-\sigma)\gamma\}$, and $\vphi_0\in H^{r}(\R^3)$. Suppose $N\in\N$ and $k\in\N^{\leq N}$. Let $\Psi_{N,t}$ be a solution of \eqref{schroedinger} and $\vphi_t$ be a solution of \eqref{hartree}. Then there is a constant $C=C_{\|\vphi_0\|_{H^{r}}}$ such that for all $t\geq0$ we have
\begin{equation*}
\begin{split}
\tr{\left|\gamma_{N,t}^{(k)}-P^{(k)}_t\right|}&\lesssim\;\sqrt{k}\frac{\e^{C\e^{Ct}}}{\sqrt{N}},\\
a_{N,t}\;&\leq\;\frac{\e^{C\e^{Ct}}}{N}.
\end{split}
\end{equation*}
If $\sigma\geq\gamma/(\gamma+1)$, we only need to assume $\vphi_0\in H^\sigma(\R^3)$ and we can improve these bounds to
\begin{equation*}
\begin{split}
\tr{\left|\gamma_{N,t}^{(k)}-P^{(k)}_t\right|}&\lesssim\;\sqrt{k}\frac{\e^{C(\sqrt{t}+t^{2})}}{\sqrt{N}},\\
a_{N,t}\;&\leq\;\frac{\e^{C(\sqrt{t}+t^{2})}}{N}
\end{split}
\end{equation*}
for some $C=C_{\|\vphi_0\|_{H^\sigma}}$.
\end{proposition}
\begin{proof}
	As explained at the end of section \ref{notations}, it is enough to bound the quantity $\|\vphi_\tau\|_{L^1_\tau([0,t],L^{6/(3-2\gamma)}_x)}$ for we then have
	$$a_{N,t}\leq\frac{\e^{c(\|\vphi_\tau\|_{L^1_\tau([0,t];L^{6/(3-2\gamma)}_x)})}}{N}$$
	for some universal constant $c$.
	\paragraph{Case $\gamma>1$.} We start by applying H\"older's inequality followed by Sobolev's inequality
	\begin{equation}
		\begin{split}
			\|\vphi_\tau\|_{L^1_\tau([0,t];L^{6/(3-2\gamma)}_x)}&\leq \sqrt{t}\|\vphi_\tau\|_{L^2_\tau([0,t];L^{6/(3-2\gamma)}_x)}\\
			&\lesssim\sqrt{t}\||\nabla|^{\gamma-1}\vphi_\tau\|_{L^2_\tau([0,t];L^6_x)}\\
			&\lesssim\sqrt{t}\|\vphi\|_{S^{\gamma-\sigma}_{2,6}}\\
			&\lesssim\sqrt{t}(\|\vphi_0\|_{H^{\gamma-\sigma}}+\|(\frac1{|\cdot|^\gamma}*|\vphi_{\tau}|^2)\vphi_\tau\|_{L^1_\tau H^{\gamma-\sigma}_x}),
		\end{split}
	\end{equation}
	where in the last step we applied Lemma \ref{fsstrichartz}. For $\sigma\geq\gamma/2>\gamma/(\gamma+1)$, we have $\gamma-\sigma\leq\gamma/2$. Thus, applying Lemma \ref{fraccontract}, we find
	\begin{equation}
	\begin{split}
		\|(\frac1{|\cdot|^\gamma}*|\vphi_{\tau}|^2)\vphi_\tau\|_{H^{\gamma-\sigma}}\lesssim \|\vphi_{\tau}\|_{H^{\frac{2\gamma-\sigma}{3}}}^3\leq C
	\end{split}
	\end{equation}
	for some $C=C_{\|\vphi_0\|_{H^{\sigma}}}$. In the last step, we applied Lemma \ref{fracpersreg}. 
	\paragraph{Case $\gamma\leq1$.} In this case we have $6/(3-2\gamma)\leq6$ and we have Strichartz estimates available. After applying H\"older's inequality, we thus apply Strichartz estimates \ref{fsstrichartz} to obtain
	\begin{equation}
	\begin{split}
	\|\vphi_\tau\|_{L^1_\tau([0,t];L^{6/(3-2\gamma)}_x)}&\leq t^{\frac{2-\gamma}2}\|\vphi_\tau\|_{L^{2/\gamma}_\tau([0,t];L^{6/(3-2\gamma)}_x)}\\
	&\leq t^{\frac{2-\gamma}2}\|\vphi_\tau\|_{S^{(1-\sigma)\gamma}_{2/\gamma,6/(3-2\gamma)}}\\
	&\lesssim t^{\frac{2-\gamma}2}(\|\vphi_0\|_{H^{(1-\sigma)\gamma}}+\|(\frac1{|\cdot|^\gamma}*|\vphi_{\tau}|^2)\vphi_\tau\|_{L^1_\tau H^{(1-\sigma)\gamma}_x}).
	\end{split}
	\end{equation}
	In the case $\sigma\geq\gamma/(\gamma+1)$, we apply Lemma \ref{fraccontract} to obtain
	\begin{equation}
		\begin{split}
			\|(\frac1{|\cdot|^\gamma}*|\vphi_{\tau}|^2)\vphi_\tau\|_{H^{(1-\sigma)\gamma}}\lesssim\|\vphi_{\tau}\|_{H^{\gamma/2}}^2\|\vphi_{\tau}\|_{H^{(1-\sigma)\gamma}}\leq C
		\end{split}
	\end{equation}
	for some constant $C=C_{\|\vphi_0\|_{H^\sigma}}$, where in the last step, we applied Lemma \ref{fracpersreg}. In the case $\sigma<\gamma/(\gamma+1)$, we apply Lemma \ref{fraccontract} again to find
	\begin{equation*}
	\begin{split}
	\|(\frac1{|\cdot|^\gamma}*|\vphi_{\tau}|^2)\vphi_\tau\|_{H^{(1-\sigma)\gamma}}\lesssim\|\vphi_{\tau}\|_{H^{\gamma/2}}^2\|\vphi_{\tau}\|_{H^{(1-\sigma)\gamma}}\lesssim C\e^{C\tau}
	\end{split}
	\end{equation*}
	for some constant $C=C_{\|\vphi_0\|_{H^{(1-\sigma)\gamma}}}$.
\end{proof}
In view of remark \ref{simplifyrest}, we obtain
\begin{proposition}\label{fssubcrit}
	Let $r:=\max\{\sigma;(1-\sigma)\gamma\}$ and $\vphi_0\in H^{r}(\R^3)$. Suppose $N\in\N$ and $k\in\N^{\leq N}$. Let $\Psi_{N,t}$ be a solution of \eqref{schroedinger} and $\vphi_t$ be a solution of \eqref{hartree}. Then there is a constant $C=C_{\|\vphi_0\|_{H^{r}}}$ such that for all $t\geq0$ and any $N\in\N$ we have
	$$\tr\left|S_{k,\theta}^{\frac12} (\gamma_{N,t}^{(k)}-P^{(k)}_t) S_{k,\theta}^{\frac12}\right|\leq Ck^{\frac{3-\theta}2}\frac{\e^{C\e^{Ct}}}{N^{\frac{1-\theta}{2}}}.$$
	If $\sigma\geq\gamma/(\gamma+1)$, we only need to require $\vphi_0\in H^\sigma(\R^3)$ and we can improve this bound to
	$$\tr\left|S_{k,\theta}^{\frac12} (\gamma_{N,t}^{(k)}-P^{(k)}_t) S_{k,\theta}^{\frac12}\right|\leq Ck^{\frac{3-\theta}2}\frac{\e^{C(\sqrt{t}+t^{2})}}{N^{\frac{1-\theta}{2}}}$$
	for some $C=C_{\|\vphi_0\|_{H^\sigma}}$.
\end{proposition}
\paragraph{The focusing case.}  In this paragraph, assume $\mu=-1$, $\lambda\geq\lambda_{S,c}$, and $\sigma=\gamma/2\neq1/2$. Due to remark \ref{simplifyrest}, we will only state the results.
\begin{proposition}\label{fspicklsupercrit}
	Let $r:=\max\{\gamma/2;(1-\gamma/2)\gamma\}$, $\vphi_0\in H^{r}(\R^3)$, and define $\varepsilon\in[0,1]$ to satisfy $(1+\varepsilon)\gamma/2=r$. Fix some $T\in(0,\infty)$ such that
	$$\nu\;=\;\sup_{|\tau|\leq T} \|\vphi_\tau\|_{H^{\frac\gamma2}}<\infty.$$
	Suppose $N\in\N$ and $k\in\N^{\leq N}$. Let $\Psi_{N,t}^{(\alpha)}=\e^{-\i H_N^{(\alpha)}t}\vphi_0^{\otimes N}$ and $\vphi$ be a solution of \eqref{hartree}. Then there is a constant $C=C_{\nu,\|\vphi_0\|_{H^{r}},T}$ such that with the above notations for any $0<\alpha<1$ and any $t\in [0,T]$ we have
	\begin{equation*}
	\begin{split}
	\tr{\left|\gamma_{N,t}^{(k,\alpha)}-P^{(k)}_t\right|}&\leq C\sqrt{k}\left(\frac{1}{\sqrt{N}}+\alpha^{\frac{1+\varepsilon}{2}}\right),\\
	a_{N,t}^{(\alpha)}&\leq\frac{C}{N}.
	\end{split}
	\end{equation*}
	In particular, we have for $\alpha=\alpha_N=\mathcal{O}(N^{-1})$
	\begin{equation*}
	\tr{\left|\gamma_{N,t}^{(k,\alpha_N)}-P^{(k)}_t\right|}\lesssim C \sqrt{\frac{k}{N}}.
	\end{equation*}
	If $\gamma>1$, we only have to assume $\vphi_0\in H^{\gamma/2}(\R^3)$ and the same estimates hold with $C=C_{\nu,\|\vphi_0\|_{H^{\gamma/2}},T}$ and $\varepsilon=0$.
\end{proposition}
\begin{remark}
	Notice the reduced rate $\alpha^{\frac{1+\varepsilon}{2}}$ due to reduced required regularity on initial data. This effectively changes the step \eqref{srregpickll2dist} in the proof of Proposition \ref{srpicklsupercrit} invoking Lemma \ref{fracregapprox}.
\end{remark}

\begin{proposition}\label{fshighersobolevconv}
	Let $r:=\max\{\gamma/2;(1-\gamma/2)\gamma\}$, $\vphi_0\in H^{r}(\R^3)$, and define $\varepsilon\in(0,1]$ to satisfy $(1+\varepsilon)\gamma/2\geq r$. Fix some $T\in(0,\infty)$ such that
	$$\nu\;=\;\sup_{|\tau|\leq T} \|\vphi_\tau\|_{H^{\frac\gamma2}}<\infty.$$
	Suppose $N\in\N$ and $k\in\N^{\leq N}$. Let $\Psi_{N,t}^{(\alpha)}=\e^{-\i H_N^{(\alpha)}t}\vphi_0^{\otimes N}$ and $\vphi$ be a solution of \eqref{hartree}. Then there is a constant $C=C_{\nu,\|\vphi_0\|_{H^{r}},T}$ such that with the above notations for any $0<\alpha<1$ small enough, any $t\in [0,T]$ and any $\theta\in[0,\min\{\frac{\vep}{1-\vep};1\})$ we have
	$$\tr\left|S_{k,\theta}^{\frac12} (\gamma_{N,t}^{(k,\alpha)}-P^{(k)}_t) S_{k,\theta}^{\frac12}\right|\leq C k^{\frac{3-\theta}2}\left( \frac1{\alpha^{\max\{1/2;\theta\}}N^{\frac{1-\theta}{2}}}+\alpha^{\frac{1-\theta+\varepsilon(1+\theta)-\max\{1;2\theta\}}{2}}\right).$$
	In particular, we have for $\alpha=\alpha_N=\mathcal{O}(N^{-(1-\theta)/(1-\theta+\varepsilon(1+\theta))})$
	$$\tr\left|S_{k,\theta}^{\frac12} (\gamma_{N,t}^{(k,\alpha)}-P^{(k)}_t) S_{k,\theta}^{\frac12}\right|\leq C\frac{k^{\frac{3-\theta}2}}{N^{\frac{(1-\theta)(1-\theta+\varepsilon(1+\theta)-\max\{1;2\theta\})}{2[1-\theta+\varepsilon(1+\theta)]}}}.$$
\end{proposition}
\begin{remark}
	As above, we have changed rates in $\alpha$ due to the adjusted initial regularity. In here, we effectively only have to change the pendants to \eqref{srprojetrace} and \eqref{srhs} according to Lemma \ref{fracregapprox}.
\end{remark}

\subsection{The non-relativistic case}
	In this section, we always assume \ref{itm:nr}. Also in this case, a pair $(q,r)\in[2,\infty]^2$ is called \textit{admissible} iff
	$$\frac2q+\frac3r=\frac32.$$ 
	Since, in the case of the full Laplacian, Strichartz estimates have been intensively studied, we will only give the reference to Tao's book \cite{Tao} and refer the reader to the references therein.
	\begin{lemma}\label{nrstrichartz}
		Let $(q,r),(\tilde{q},\tilde{r})\in[2,\infty]^2$ be admissible pairs and $s\in\R$. Then we have for any $I\subseteq\R$
		\begin{equation*}
			\begin{split}
				\|\e^{-\i (-\Delta)t}\vphi_0\|_{L^q_t(I;W^{s,r}_x)}\;&\lesssim\;\|\vphi_0\|_{H^s},\\
				\|\int_0^t \e^{-\i(-\Delta)(t-\tau)}F(\tau)\mathrm{d}\tau\|_{L^q_t(I;W^{s,r}_x)}\;&\lesssim\;\|F\|_{L^{\tilde{q}'}_t(I;W^{s,\tilde{r}'}_x)}.
			\end{split}
		\end{equation*}
	\end{lemma}
	\begin{proposition}\label{nrsupconv}
		Assume $\gamma\in(0,3/2)$, and let $s:=(\gamma-1)_+$ and $\vphi_0\in H^{s}(\R^3)$. Suppose $N\in\N$ and $k\in\N^{\leq N}$. Let $\Psi_{N,t}$ be a solution of \eqref{schroedinger} and $\vphi_t$ be a solution of \eqref{hartree}. Then there is a constant $C=C_{\|\vphi_0\|_{H^s}}$ such that for all $t\geq0$ we have
		\begin{equation}
		\begin{split}
		\tr{\left|\gamma_{N,t}^{(k)}-P^{(k)}_t\right|}\;&\lesssim\;\sqrt{k}\frac{\e^{C(\sqrt{t}+t^6)}}{\sqrt{N}},\\
		a_{N,t}\;&\leq\;\frac{\e^{C(\sqrt{t}+t^6)}}{N}.
		\end{split}			
		\end{equation}
	\end{proposition}
	\begin{proof}
		As explained above, it suffices to bound the quantity $\|\vphi_\tau\|_{L^1_\tau([0,t],L^{6/(3-2\gamma)}_x)}$.
		\paragraph{Case $\gamma>1$.} Using H\"older's inequality, followed by the Sobolev inequality and the Strichartz estimate in this case, we have
		\begin{align*}
			\|\vphi_\tau\|_{L^1_\tau([0,t],L^{6/(3-2\gamma)}_x)}&\leq \sqrt{t}\|\vphi_\tau\|_{L^2_\tau([0,t],L^{6/(3-2\gamma)}_x)}\\
			&\lesssim \sqrt{t}\|\vphi_\tau\|_{L^2_\tau([0,t],W^{\gamma-1,6}_x)}\\
			&\lesssim\sqrt{t}(\|\vphi_0\|_{H^{\gamma-1}}+\|(\frac{1}{|\cdot|^\gamma}*|\vphi_{\tau}|^2)\vphi_{\tau}\|_{L^{2/\gamma}_\tau W^{\gamma-1,6/(7-2\gamma)}_x})\\
			&\lesssim\sqrt{t}(\|\vphi_0\|_{H^{\gamma-1}}+\|\vphi_{\tau}\|_{L^{2/\gamma}_\tau H^{\gamma-1}_x}^3)\\
			&\leq(\sqrt{t}+t^6)\|\vphi_0\|_{H^{\gamma-1}}^3,
		\end{align*}		
		where in the next-to-last step we applied Lemma \ref{fraccontract}.
		\paragraph{Case $\gamma\leq1$.} Then we apply H\"older's inequality followed by the Strichartz estimate \ref{nrstrichartz} to obtain
		\begin{align*}
		\|\vphi_\tau\|_{L^1_\tau([0,t],L^{6/(3-2\gamma)}_x)}&\leq t^{\frac{2-\gamma}2}(\|\vphi_0\|_2+\|(\frac{1}{|\cdot|^\gamma}*|\vphi_{\tau}|^2)\vphi_{\tau}\|_{L^{2/(2-\gamma)}_\tau L^{6/(3-2\gamma)}_x})\\
		&\lesssim \sqrt{t}+t^6,
		\end{align*}
		where in the last step we used mass conservation together with Lemma \ref{fraccontract}.
	\end{proof}
	With the needed Sobolev exponent $s$ being below $1$, we can state stronger result for convergence in higher Sobolev trace norms than the ones in the previous cases. The following result shows how, with this method, reducing imposed regularity yields slower convergence rates.
	\begin{proposition}\label{nrhighersobolev}
		Assume $\gamma\in(0,3/2)$ and $s\in[(\gamma-1)_+,1]$, and let $\vphi_0\in H^{s}(\R^3)$. Suppose $N\in\N$ and $k\in\N^{\leq N}$. Let $\Psi_{N,t}$ be a solution of \eqref{schroedinger} and $\vphi_t$ be a solution of \eqref{hartree}. Then there is a constant $C=C_{\|\vphi_0\|_{H^s}}$ such that for any $\theta\in[0,s)$ and any $t\geq0$ we have
		$$\tr\left|S_{k,s\theta}^{\frac12} (\gamma_{N,t}^{(k)}-P^{(k)}_t) S_{k,s\theta}^{\frac12}\right|\leq Ck^{\frac{3-\theta}2}\frac{\e^{C(\sqrt{t}+t^6)}}{N^{\frac{s-\theta}{2}}}.$$
	\end{proposition}
	\begin{proof}
		In view of Proposition \ref{oldmainthm}, \ref{nrsupconv} and \eqref{picklschmidt}, it is sufficient to provide uniform bounds on $\|S_{1,s}^{1/2}\Psi_{N,t}\|_2 +\|S^{s/2}\vphi_t\|_2$. By Lemma \ref{frlowregwellposed}, we have
		$$\|S^{s/2}\vphi_t\|_2\leq(t+1)\|\vphi_0\|_{H^s}.$$
		Moreover, using that $(1+S_i)^s\leq 1+S_i^s$ in the sense of quadratic forms, we find
		\begin{equation}\label{nrpsibound}
		\begin{split}
		\|S_{1,s}^{1/2}\Psi_{N,t}\|_2^2-1
		&\leq\frac1N\scp{\Psi_{N,t}}{\sum_{i=1}^NS_i^s\Psi_{N,t}}
		\end{split}
		\end{equation}
		To continue, let us recall the following facts. Let $a_1,\ldots,a_N\geq0$. Then we have
		\begin{equation}\label{algfact}
		\sum_{i=1}^Na_i^{s}\,\leq\,N^{1-s}\left(\sum_{i=1}^Na_i\right)^{s}\,\leq\,N^{1-s}\sum_{i=1}^Na_i^{s}.
		\end{equation}
		The second inequality follows, e.g., from the embedding $\ell^1(\{1;\ldots;N\})\hookrightarrow\ell^2(\{1;\ldots;N\})$ together with interpolation. By Plancherel's theorem, we can replace the $a_i$ in this inequality by operators $A_i$ with non-negative symbols $a_i$. Moreover, the L\"owner-Heinz inequality together with \eqref{hamsupercrit} implies
		\begin{equation}\label{hamsupercrit2s}
			\left(\sum_{i=1}^NS_i+N\right)^s\lesssim (H_N+N)^s\lesssim\left(\sum_{i=1}^NS_i+N\right)^s.
		\end{equation}
		Using these two facts on \eqref{nrpsibound} together with symmetry of $\Psi_{N,t}$ w.r.t. to particle permutations, we obtain using first \eqref{algfact}, then \eqref{hamsupercrit2s}, then energy conservation for the Schr\"odinger equation, then \eqref{hamsupercrit2s}, and then \eqref{algfact} again
		\begin{align*}
			\|S_{1,s}^{1/2}\Psi_{N,t}\|_2^2&=\frac1N\scp{\Psi_{N,t}}{\sum_{i=1}^N(S_i+1)^s\Psi_{N,t}}\\
			&\leq\frac1{N^s}\scp{\Psi_{N,t}}{\left(\sum_{i=1}^NS_i+N\right)^s\Psi_{N,t}}\\
			&\lesssim\frac1{N^s}\scp{\Psi_{N,t}}{(H_N+N)^s\Psi_{N,t}}\\
			&=\frac1{N^s}\scp{\Psi_{N,0}}{(H_N+N)^s\Psi_{N,0}}\\
			&\lesssim\frac1{N^s}\scp{\Psi_{N,0}}{\left(\sum_{i=1}^NS_i+N\right)^s\Psi_{N,0}}\\
			&\leq \frac1{N^s}\scp{\Psi_{N,0}}{\sum_{i=1}^N(S_i+1)^s\Psi_{N,0}}\\
			&=N^{1-s}\|\vphi_0\|_{H^s}^2,
		\end{align*}
		where in the last step we also used $\Psi_{N,0}=\vphi_0^{\otimes N}$. Together with the initial comments, this finishes the proof.
	\end{proof}
We can improve the rates in the last theorem recalling the result
$$\|\gamma_{N,t}^{(1)}-P^{(1)}_t\|_{HS}\leq\frac{C\e^{Ct^{3/2}}}{N}$$
given in \cite{CLL}, where $C=C_{\|\vphi_0\|_{H^1}}$ and $\vphi_0\in H^1(\R^3)$. Combining this result with Proposition \ref{oldmainthm}, Proposition \ref{nrsupconv}, Lemma \ref{fracpersreg}, and the proof of Proposition \ref{nrhighersobolev} in the case $s=1$, we get the following the result.
\begin{corollary}
		Assume $\vphi_0\in H^{1}(\R^3)$ and $N\in\N$. Let $\Psi_{N,t}$ be a solution of \eqref{schroedinger} and $\vphi_t$ be a solution of \eqref{hartree}. Then there is a constant $C=C_{\|\vphi_0\|_{H^1}}$ such that for any $\theta\in[0,1)$ and any $t\geq0$ we have
		$$\tr\left|S_{1,\theta}^{\frac12} (\gamma_{N,t}^{(1)}-P^{(1)}_t) S_{1,\theta}^{\frac12}\right|\leq C\frac{\e^{C(\sqrt{t}+t^6)}}{N^{\min\{1/2;1-\theta\}}}.$$
\end{corollary}

\appendix

\section{Well-posedness results\label{fracpersproof}}
\subsection{Energy conservation and self-adjointness of $H_N$}
Let us first prove Lemma \ref{fracconservation}. Define
\begin{align*}
J_\gamma(u)\;&:=\;\frac{\mu\lambda}{|\cdot|^\gamma}*|u|^2,\\
J_\gamma^{(\alpha)}(u)\;&:=\;\frac{\mu\lambda}{|\cdot|^\gamma+\alpha}*|u|^2.
\end{align*}

\begin{proof}
	Conservation of mass follows from differentiating the (real) mass and noticing that the r.h.s. is pure imaginary. For the conservation of energy, we work with the interaction picture, as noticed in \cite{AHH}: Starting with the identity
	$$\partial_t(\e^{\i (-\Delta)^\sigma t}\vphi_t)=-\i\e^{\i (-\Delta)^\sigma t}J_\gamma(\vphi_t)\vphi_t,$$
	we employ the fact that $\e^{\i(-\Delta)^\sigma t}$ is an isometric embedding of any homogeneous Sobolev space to obtain
	\begin{align*}
	\partial_t T^{(\sigma)}[\vphi_t]&=\partial_tT^{(\sigma)}[\e^{\i(-\Delta)^\sigma t}\vphi_t]\\&=\mathrm{Re}\scp{(-\Delta)^{\frac\sigma2}\e^{\i(-\Delta)^\sigma t}\vphi_t}{-\i(-\Delta)^{\frac\sigma2}\e^{\i(-\Delta)^\sigma t}J_\gamma(\vphi_t)\vphi_t}\\
	&=\mathrm{Re}\scp{(-\Delta)^{\frac\sigma2}\vphi_t}{-\i(-\Delta)^{\frac\sigma2}J_\gamma(\vphi_t)\vphi_t}.
	\end{align*}
	On the other hand, we have
	\begin{align*}
	\partial_t V^{(\gamma)}[\vphi_t]&=\mathrm{Re}\scp{\partial_t \vphi_t}{J_\gamma(\vphi_t)\vphi_t}_{(H^{-\frac{\sigma}2};H^{\frac{\sigma}2})}\\&=\mathrm{Re}\scp{-\i(-\Delta)^{\sigma}\vphi_t-\i J_\gamma(\vphi_t)\vphi_t}{J_\gamma(\vphi_t)\vphi_t}_{(H^{-\frac{\sigma}2};H^{\frac{\sigma}2})}\\
	&=\mathrm{Re}\scp{(-\Delta)^{\frac{\sigma}2}\vphi_t}{\i(-\Delta)^{\frac{\sigma}2}J_\gamma(\vphi_t)\vphi_t}.
	\end{align*}
	Adding the derivative of both the kinetic and the potential energy, we have proved $\partial_t\E=0$. The same arguments hold if we replace $J_\gamma$ by $J_\gamma^{(\alpha)}$.
\end{proof}

\begin{proof}[Proof of Lemma \ref{hnsa}]
	Let us recall some facts from interpolation. If we define the weight $\omega:=|\cdot|^{-2}$, we find the interpolation space
	$$[L^2,L^2_\omega]_{\theta}=L^2_{\omega^\theta},\quad0<\theta<1,$$
	where $L^2_\omega=L^2_\omega(\R^3)$ denotes the weighted $L^2$-space with weight $\omega$. Next, we have
	$$[L^2,\dot{H}^1]_\theta=\dot{H}^\theta,\quad0<\theta<1,$$
	where $\dot{H}^\theta=\dot{H}^\theta(\R^3)$ denotes the homogeneous Sobolev space. Hardy's inequality yields an embedding $\dot{H}^1\hookrightarrow L^2_\omega$, which in turn together with the trivial embedding $L^2\hookrightarrow L^2$ yields by interpolation the embedding
	$$\dot{H}^\theta\hookrightarrow L^2_{\omega^\theta}.$$
	In the case $\theta=\sigma=\gamma/2$, this corresponds to an estimate $\frac1{|\cdot|^\gamma}\lesssim(-\Delta)^{\gamma/2}$ in the sense of quadratic forms. In the mass subcritical regime $\sigma>\gamma/2$, we employ the trivial embedding $H^{\gamma/2}\hookrightarrow \dot{H}^{\gamma/2}$ together with the previous embedding and the interpolation inequality 
	$$\|u\|_{H^{\gamma/2}}\;\lesssim\;\|u\|_{H^{\sigma}}^{\frac{\gamma}{2\sigma}}\|u\|_2^{1-\frac{\gamma}{2\sigma}}$$
	to find
	$$\||\cdot|^{-\gamma/2}u\|_2\lesssim \|u\|_{H^{\sigma}}^{\frac{\gamma}{2\sigma}}\lesssim \varepsilon\|u\|_{H^\sigma}+C_\varepsilon$$
	if $\|u\|_2=1$ is assumed and $\varepsilon>0$ is chosen arbitrarily small. 
\end{proof}

\subsection{High Sobolev regularity}
Since the well-posedness for $\sigma=1$ is well-known, we will, as above, closely follow the analysis in \cite{Le}. For that, recall the following generalized Leibniz rule, see \cite{Le}, which itself is an easy consequence of the Mihilin multiplier theorem. For some $\mu\geq0$, we denote $\mathcal{D}:=\sqrt{-\Delta}$.
\begin{lemma}[Leibniz-rule]\label{leibniz}
	Suppose that $1<p<\infty$, $s\geq0$, $a,b\geq0$, and for $i=1,2$ let $1/p_i+1/q_i=1$ with $1<p_i,q_i\leq\infty$. Then
	$$\|\mathcal{D}^s(fg)\|_p\lesssim \|\mathcal{D}^{s+a}f\|_{p_1}\|\mathcal{D}^{-a}g\|_{q_1}+\|\mathcal{D}^{-b}f\|_{p_2}\|\mathcal{D}^{s+b}g\|_{q_2},$$
	where the involved constant depends on all parameters but not on $f$ and $g$.
\end{lemma}
The following lemma ensures that both $u\mapsto J_\gamma(u)u$ and $u\mapsto J_\gamma^{(\alpha)}(u)u$ locally Lipschitz.
\begin{lemma}\label{fraccontract}
	Let $s\geq0$. If $u,v\in H^{\max\{\gamma/2;s\}}(\R^3)$, we have 
	\begin{equation*}
	\begin{split}
	\|J_\gamma(u)u-J_\gamma(v)v\|_{H^s}&\lesssim(\|u\|_{H^s}^2+\|v\|_{H^s}^2)\|u-v\|_{H^s},\\
	\|J_\gamma(u)u\|_{H^s}&\lesssim \|u\|_{6/(3-\gamma)}^2\|u\|_{H^s}\lesssim \|u\|_{H^{\gamma/2}}^2\|u\|_{H^s}.
	\end{split}
	\end{equation*}
	Moreover, if $s\leq\gamma/2$, $p\in[\frac6{2\gamma-4s+3},2]$, and $w\in H^{\frac{\gamma+s}3+\frac12-\frac1p}(\R^3)$, we have 
	$$\|J_\gamma(w)w\|_{W^{s,p}}\lesssim \|w\|_{H^{\frac{\gamma+s}3+\frac12-\frac1p}}^3.$$
	The same results hold true if we replace $J_\gamma$ by $J_\gamma^{(\alpha)}$.
\end{lemma}
\begin{remark}
	The second statement equally holds if we replace the Sobolev spaces by homogeneous Sobolev spaces.
\end{remark}
\begin{proof}
	We will only prove the second and the third inequality in the statement. As presented above for the magnetic case, one can show using similar estimates that $u\mapsto J(u)u$ indeed is a locally Lipschitz map from $H^s$ into itself. In addition, we will show the estimates on the homogeneous part of the respective Sobolev space, since the $L^2$-part follows by formally setting $s=0$. By Lemma \ref{leibniz}, we find
	\begin{align*}
	\|\D^s(J_\gamma(u)u)\|_2&\lesssim \|\D^sJ_\gamma(u)\|_p\|u\|_q+\|\frac{\mu\lambda}{|\cdot|^\gamma}*|u|^2\|_\infty\|\D^su\|_2\\
	&\lesssim \|\D^sJ_\gamma(u)\|_p\|u\|_q+\|u\|_{6/(3-\gamma)}^2\|\D^su\|_2\\
	&\lesssim \|\D^sJ_\gamma(u)\|_p\|u\|_q+\|u\|_{H^{\gamma/2}}^2\|\D^su\|_2,
	\end{align*}
	where also applied Young's inequality followed by the Sobolev embedding. $1<p,q\leq\infty$ satisfy $1/p+1/q=1/2$. Subsequently, we use the identity
	$$\D^{r-3}f\;=\;\frac{c_r}{|\cdot|^r}*f$$
	for some $c_r\in\R$. Also note, that $|\cdot|^{-\alpha}\in L^{3/(3-\alpha)}_w(\R^3)$.
	\paragraph{Case $s<3/2$:} Choose $p=3/s$ and $q=6/(3-2s)$.  Using the Leibniz-rule \ref{leibniz} followed by the Sobolev embedding and the weak Young's inequality, we find
	\begin{align*}
	\|\D^sJ_\gamma(u)\|_{3/s}&\lesssim\|\D^{s+\gamma-3}|u|^2\|_{3/s}\\
	&\lesssim\|\D^{s+\gamma/2-3/2}u\|_{3/s}\|\D^{\gamma/2-3/2}u\|_\infty\\
	&\lesssim\|\D^{\gamma/2}u\|_2\|u\|_{6/(3-\gamma)}\\
	&\lesssim \|u\|_{H^{\gamma/2}}^2,
	\end{align*}
	where in the last step we applied Sobolev's inequality again. Using the Sobolev embedding $\|u\|_{6/(3-2s)}\lesssim \|u\|_{H^s}$ finishes this case.
	\paragraph{Case $s\geq3/2$:} Choose $q=6/(3-\gamma)$ and $p=6/\gamma$. Using the Sobolev embedding followed by Lemma \ref{leibniz} and the weak Young's inequality, we have
	\begin{equation*}
	\begin{split}
	\|\D^sJ_\gamma(u)\|_{6/\gamma}&\lesssim\|\D^{s+\gamma/2-3/2}|u|^2\|_2\\
	&\lesssim \|\D^su\|_2\|\D^{\gamma/2-3/2}u\|_\infty\\
	&\lesssim \|u\|_{H^s}\|u\|_{6/(3-\gamma)}\\
	&\lesssim \|u\|_{H^s}\|u\|_{H^{\gamma/2}},
	\end{split}			
	\end{equation*}
	where in the last step we used Sobolev's inequality again. One further application of Sobolev's inequality yields
	$$\|\D^s(J_\gamma(u)u)\|_2\lesssim \|u\|_{H^{\gamma/2}}^2\|u\|_{H^s}.$$
	\paragraph{Case $s\leq\gamma/2$, $p\in[\frac6{2\gamma-4s+3},2]$.} Note that we have 
	$$\frac{\gamma+s}{3}+\frac12-\frac1p\geq0.$$
	Rather than the above estimate, we apply the Leibniz rule followed by the weak Young's inequality and Sobolev's inequality to obtain
	\begin{align*}
	\|\D^s(J_\gamma(w)w)\|_p&\lesssim\|\D^{s+\gamma-3}|w|^2\|_{\frac{9p}{6+(\gamma+s)p-3p}}\|w\|_{\frac{9p}{3-(\gamma+s)p+3p}}+\|\D^{\frac{2(\gamma+s)}3-\frac72+\frac1p}|w|^2\|_{\frac{2p}{2-p}}\|\D^{\frac{\gamma+s}3+\frac12-\frac1p}w\|_2\\
	&\lesssim\|w\|_{\frac{9p}{3-(\gamma+s)p+3p}}^2\|\D^{\frac{\gamma+s}3+\frac12-\frac1p}w\|_2\\
	&\lesssim\|\D^{\frac{\gamma+s}3+\frac12-\frac1p}w\|_2^3,
	\end{align*}
	where in the last step we applied the Sobolev embedding again.
	\par To obtain the results for $J_\gamma^{(\alpha)}$, observe that $$\|\cD^s J_\gamma^{(\alpha)}(u)\|_p\lesssim \|\frac1{|\cdot|^\gamma}|\cD^s|u|^2\|_p$$
	for any $s\geq0$, $p\geq1$. With that, we can repeat all above arguments.
\end{proof}
With this result, let us prove the well-posedness of \eqref{hartree} first as follows: The previous lemma shows that the non-linearity $u\mapsto J(u)u$ is locally Lipschitz from $H^s$ into itself and we can run a fixed point argument. \par Next, let us prove Lemma \ref{fracpersreg}. For that, fix a time $T>0$ such that
$$\nu:=\sup_{t\in[0,T]}\|\vphi_t\|_{H^{\gamma/2}}<\infty.$$
Since $(-\Delta)^\sigma$ is self-adjoint, it induces a $C_0$-group $(\e^{-\i (-\Delta)^\sigma t})_{t\in\R}$ of isometries. Thus, using Duhamel's formula
$$\vphi_t\;=\;\e^{-\i (-\Delta)^\sigma t}\vphi_0-\i\int_0^t \e^{-\i (-\Delta)^\sigma(t-\tau)}(\frac{\mu\lambda}{|\cdot|^\gamma}*|\vphi_\tau|^2)\vphi_\tau \mathrm{d}\tau$$ 
together with the last lemma, we find
$$\|\vphi_t\|_{H^s}\lesssim \|\vphi_0\|_{H^s}+\nu^2\int_0^t\|\vphi_\tau\|_{H^s}\mathrm{d}\tau.$$
Gronwall's inequality then yields the first part of Lemma \ref{fracpersreg}. If $\mu=1$ or $\lambda<\lambda_{H,c}$, and $s=\sigma$, we have due to Lemmata \ref{fracconservation} and \ref{energypos} that
$$\|\vphi_t\|_{H^\sigma}\lesssim\|\vphi_0\|_{H^\sigma}.$$
The remaining cases $s\leq\sigma$ follow from interpolation and this last estimate and mass conservation. This finishes the proof of Lemma \ref{fracpersreg}. 
\par In order to show global well-posedness in the case $\mu=1$ or $\lambda<\lambda_{H,c}$, we employ Lemma \ref{fracpersreg} to obtain
\begin{equation*} 
\|\vphi_t\|_{H^s}\leq \e^{c\|\vphi_0\|_{H^\sigma}^{\frac{\gamma}{\sigma}}t} \|\vphi_0\|_{H^s}
\end{equation*}
for some universal constant $c$ if $s\geq\sigma$. Thus, by Picard iteration, we obtain that \eqref{hartree} is globally well-posed. For the finite-time blow-up when $\mu=-1$ and $\lambda>\lambda_{H,c}$, see, e.g., \cite{Zh}
	
	\subsection{Low Sobolev regularity}
	
	Since the arguments run similarly to the ones in the previous subsection, we will only mention the main steps here. We start with Duhamel's formula which reads
	$$\vphi_t\;=\;\e^{-\i(-\Delta)t}\vphi_0-\i\int_0^t\e^{-\i(-\Delta)(t-\tau)}(\frac{\mu\lambda}{|\cdot|^\gamma}*|\vphi_\tau|^2)\vphi_\tau\mathrm{d}\tau.$$
	Fix some time $T>0$. Similarly to the above, we start by employing that $\e^{-\i(-\Delta)t}$ is an $L^2$-isometry together with Lemma \ref{fraccontract} to obtain
	\begin{equation}\label{nrslg1}
	\begin{split}
	\|\vphi_t\|_{H^s}&\lesssim \|\vphi_0\|_{H^s}+\int_0^t\|\vphi_\tau\|_{6/(3-\gamma)}^2\|\vphi_\tau\|_{H^s}\mathrm{d} \tau\\
	&\lesssim\|\vphi_0\|_{H^s}+\|\vphi_{\tau}\|_{L^\infty_\tau([0,t];H^{s}_x)}\int_0^t\|\vphi_\tau\|_{6/(3-\gamma)}^2\mathrm{d}\tau.
	\end{split}
	\end{equation}
	Next we can employ Strichartz estimates following H\"older's inequality to obtain
	\begin{equation}\label{nrslg2}
	\begin{split}
	\|\vphi_{\tau}\|_{L^2_\tau([0,t];L^{6/(3-\gamma)}_x)}&\leq t^{\frac4{2-\gamma}} \|\vphi_{\tau}\|_{L^{4/\gamma}_\tau([0,t];L^{6/(3-\gamma)}_x)}\\
	&\lesssim t^{\frac4{2-\gamma}}\left(1+\|(\frac{\mu\lambda}{|\cdot|^\gamma}*|\vphi_\tau|^2)\vphi_\tau\|_{L^{2/(2-\gamma)}_\tau([0,t]; L^{6/(2\gamma+3)}_x)}\right)\\
	&\lesssim t^{\frac4{2-\gamma}},
	\end{split}
	\end{equation}
	where in the last step we applied Lemma \ref{fraccontract} together with mass conservation. Inserting \eqref{nrslg2} into \eqref{nrslg2} implies
	\begin{equation}\label{nrglobalslg}
	\|\vphi_t\|_{L^\infty_t([0,T];H^s_x)}\lesssim \|\vphi_0\|_{H^s}+T^{\frac4{2-\gamma}}	\|\vphi_t\|_{L^\infty_t([0,T];H^s_x)}.
	\end{equation} 
	A small caveat at this point is that we cannot directly apply Lemma \ref{fraccontract} for $\vphi_0$ is not necessarily in $H^{\gamma/2}(\R^3)$. Instead we have to use an improved version of Lemma \ref{fraccontract} involving Strichartz estimates as shown above. The details are left to the reader. 
	\par With similar estimates, one can show, that for $T>0$ small enough $u\mapsto J_\gamma(u)u$ is a locally Lipschitz map of $C([0,T];H^s(\R^3))$ into itself. Then one can run a standard fixed point argument to show local well-posedness. For global well-posedness, notice that in \eqref{nrglobalslg} we can choose $T>0$ independently of $\vphi_0$ so small that 
	$$\|\vphi_t\|_{L^\infty_t([0,T];H^s_x)}\lesssim \|\vphi_0\|_{H^s}.$$
	Iterating this inequality yields
	\begin{equation*}
	\|\vphi_\tau\|_{L^\infty_\tau([0,t];H^s_x)}\lesssim \lceil\frac{t}{T}\rceil \|\vphi_0\|_{H^s}\lesssim (t+1)\|\vphi_0\|_{H^s}.
	\end{equation*}
	This completes the proof of global well-posedness.

	\section{Approximating the Hartree equation by its regularized version \label{hartreeregapprox}}
	In here, we want to show a stronger statement than that of Lemma \ref{fracregapprox} in that we allow for any values $\gamma\in(0,3/2)$ and $\sigma\in[\gamma/2,1]$. Our idea is to follow the steps of \cite{MS} and modify them suitably. We start by proving the second estimate. Using Duhamel's formula for both $\vphi_t$ and $\vphi_t^{(\alpha)}$, we find
	\begin{equation}\label{gradestsplit}
	\begin{split}
		|(-\Delta)^{\frac{\gamma}{4}}(\vphi_t-\vphi_t^{(\alpha)})\|_2\lesssim\int_0^t\mathrm{d}\tau&\left\{\|(-\Delta)^{\frac{\gamma}{4}}(\frac1{|\cdot|^\gamma}*|\vphi_{\tau}|^2)(\vphi_\tau-\vphi_{\tau}^{(\alpha)})\|_2\right.\\
		&+\|(-\Delta)^{\frac{\gamma}{4}}(\frac\alpha{|\cdot|^\gamma(|\cdot|^\gamma+\alpha)}*|\vphi_{\tau}|^2)(\vphi_\tau-\vphi_{\tau}^{(\alpha)})\|_2\\
		&+\|(-\Delta)^{\frac{\gamma}{4}}(\frac\alpha{|\cdot|^\gamma(|\cdot|^\gamma+\alpha)}*|\vphi_{\tau}|^2)\vphi_\tau\|_2\\
		&+\|(-\Delta)^{\frac{\gamma}{4}}(\frac1{|\cdot|^\gamma+\alpha}*(|\vphi_{\tau}|^2-|\vphi_{\tau}^{(\alpha)}|^2))\vphi_\tau\|_2\\
		&\left.+\|(-\Delta)^{\frac{\gamma}{4}}(\frac1{|\cdot|^\gamma+\alpha}*(|\vphi_{\tau}|^2-|\vphi_{\tau}^{(\alpha)}|^2))(\vphi_\tau-\vphi_{\tau}^{(\alpha)})\|_2 \right\}.
	\end{split}
	\end{equation}
	The first term, we can estimate using the generalized Leibniz rule \ref{leibniz} together with the fact that $|\nabla|^{\alpha-3}u=|\cdot|^{-\alpha}*u$ by
	\begin{equation}\label{gradestsplit1}
		\begin{split}
			\|\frac1{|\cdot|^{3\gamma/2}}*|\vphi_{\tau}|^2\|_{6/\gamma}\|\vphi_{\tau}-\vphi_{\tau}^{(\alpha)}\|_{6/(3-\gamma)}+\|\frac1{|\cdot|^\gamma}*|\vphi_{\tau}|^2\|_\infty\|(-\Delta)^{\frac\gamma4}(\vphi_{\tau}-\vphi_{\tau}^{(\alpha)})\|_2\\
			\lesssim \nu^2\|(-\Delta)^{\frac\gamma4}(\vphi_{\tau}-\vphi_{\tau}^{(\alpha)})\|_2,
		\end{split}
	\end{equation}
	where in the second step we applied the weak Young's inequality followed by the Sobolev embedding.
	\par For the second term of \eqref{gradestsplit}, we also apply the Leibniz rule followed by the weak Young's inequality and the Sobolev inequality to find the upper bound
	\begin{equation}\label{gradestsplit2aid}
		\begin{split}
			\|\frac{1}{|\cdot|^{\gamma}}*|(-\Delta)^{\frac\gamma4}|\vphi_{\tau}|^2|\|_{\frac6\gamma}\|\vphi_{\tau}-\vphi_{\tau}^{(\alpha)}\|_{\frac6{3-\gamma}}+\|\frac{1}{|\cdot|^{\gamma}}*|\vphi_{\tau}|^2\|_{\infty}\|(-\Delta)^{\frac\gamma4}(\vphi_{\tau}-\vphi_{\tau}^{(\alpha)})\|_2\\
			\lesssim(\|(-\Delta)^{\frac\gamma4}|\vphi_{\tau}|^2\|_{\frac{6}{6-\gamma}}+\|\vphi_{\tau}\|_{\frac{6}{3-\gamma}}^2)\|(-\Delta)^{\frac\gamma4}(\vphi_{\tau}-\vphi_{\tau}^{(\alpha)})\|_2,
		\end{split}
	\end{equation}
	where we also employed the fact that the fractional Laplacian is translation invariant. Applying the Leibniz rule again followed by the Sobolev embedding implies
	$$\|(-\Delta)^{\frac\gamma4}|\vphi_{\tau}|^2\|_{\frac{6}{6-\gamma}}\lesssim\|(-\Delta)^{\frac\gamma4}\vphi_{\tau}\|_{\frac{6}{6-\gamma}}\|\vphi_{\tau}\|_{\frac{6}{3-\gamma}}\lesssim\|(-\Delta)^{\frac\gamma4}\vphi_{\tau}\|_2^2\lesssim\nu^2.$$
	Together with \eqref{gradestsplit2aid}, we find by applying the Sobolev embedding again that the second term of \eqref{gradestsplit} can be estimated by
	\begin{equation}\label{gradestsplit2}
		\nu^2\|(-\Delta)^{\frac\gamma4}(\vphi_{\tau}-\vphi_{\tau}^{(\alpha)})\|_2.
	\end{equation}
	For the third term of \eqref{gradestsplit}, we use the Leibniz rule together with the translation invariance of the Laplacian to obtain the upper bound
	\begin{equation}\label{gradestsplit3aid}
		\begin{split}
			\|\frac{\alpha^\varepsilon}{|\cdot|^{(1+\varepsilon)\gamma}}*|(-\Delta)^{\frac\gamma4}|\vphi_{\tau}|^2|\|_{\frac6\gamma}\|\vphi_{\tau}\|_{\frac6{3-\gamma}}+\|\frac{\alpha^\varepsilon}{|\cdot|^{(1+\varepsilon)\gamma}}*|\vphi_{\tau}|^2\|_{\infty}\|(-\Delta)^{\frac\gamma4}\vphi_{\tau}\|_2\\	\lesssim \alpha^\varepsilon(\|(-\Delta)^{\frac\gamma4}|\vphi_{\tau}|^2\|_{\frac{6}{6-(1+2\varepsilon)\gamma}}+\|\vphi_{\tau}\|_{\frac{6}{3-(1+\gamma)\varepsilon}}^2)\nu,
		\end{split}
	\end{equation}
	where we also employed the fact that the fractional Laplacian is translation invariant. Applying the Leibniz rule again together with the fact that $|\nabla|^{\alpha-3}u=|\cdot|^{-\alpha}*u$ on the first term, we find
	\begin{equation*}
	\begin{split}
	\|(-\Delta)^{\frac\gamma4}|\vphi_{\tau}|^2\|_{\frac{6}{6-(1+2\varepsilon)\gamma}}&\lesssim \|(-\Delta)^{(1+\varepsilon)\frac{\gamma}{4}}\vphi_{\tau}\|_2\|\frac1{|\cdot|^{3-\varepsilon\gamma/2}}*\vphi_{\tau}\|_{\frac6{3-(1+2\varepsilon)\gamma}}\\
	&\lesssim \|(-\Delta)^{(1+\varepsilon)\frac{\gamma}{4}}\vphi_{\tau}\|_2\|\vphi_\tau\|_{\frac6{3-(1+\varepsilon)\gamma}},\\
	&\lesssim\|(-\Delta)^{(1+\varepsilon)\frac{\gamma}{4}}\vphi_{\tau}\|_2^2,
	\end{split}
	\end{equation*}
	where we also applied the weak Young's inequality followed by the Sobolev embedding. By Lemma \ref{fracpersreg}, we have
	$$\|\vphi_{\tau}\|_{H^{(1+\varepsilon)\gamma/2}}\lesssim\e^{c\nu^2\tau}\|\vphi_0\|_{H^{(1+\varepsilon)\gamma/2}}.$$
	The last two inequalities together with \eqref{gradestsplit3aid}, imply after applying the Sobolev embedding again that the third term of \eqref{gradestsplit} can be estimated by
	\begin{equation}\label{gradestsplit3}
	\alpha^\varepsilon\e^{c\nu^2\tau}\|\vphi_0\|_{H^{(1+\varepsilon)\gamma/2}}^2\nu.
	\end{equation}
	Note that in the case $(1+\varepsilon)\gamma/2\leq\sigma$ and $\lambda<\lambda_{H,c}$ or $\mu=1$, Lemma \ref{fracpersreg} allows us to improve this bound to a time-independent one.
	For an upper bound of the fourth term, we again apply the generalized Leibniz rule to get
	\begin{equation}\label{gradestsplit4aid}
	\begin{split}
	\|\frac1{|\cdot|^\gamma}*\left|(-\Delta)^{\frac\gamma4}(|\vphi_{\tau}|^2-|\vphi_\tau^{(\alpha)}|^2)\right|\|_{\frac6\gamma}\|\vphi_{\tau}\|_{\frac6{3-\gamma}}+	\|\frac1{|\cdot|^\gamma}*||\vphi_{\tau}|^2-|\vphi_\tau^{(\alpha)}|^2|\|_{\infty}\|(-\Delta)^{\frac{\gamma}4}\vphi_{\tau}\|_2.\\
	\lesssim(\|(-\Delta)^{\frac\gamma4}(|\vphi_{\tau}|^2-|\vphi_\tau^{(\alpha)}|^2)\|_{\frac6{6-\gamma}}+\||\vphi_{\tau}|^2-|\vphi_\tau^{(\alpha)}|^2\|_{\frac3{3-\gamma}})\|(-\Delta)^{\frac\gamma4}\vphi_{\tau}\|_2
	\end{split}	
	\end{equation}
	Next we use the fact that we can write
	$$|\vphi_{\tau}|^2-|\vphi_{\tau}^{(\alpha)}|^2\;=\;(\vphi_{\tau}-\vphi_{\tau}^{(\alpha)})(2\bar{\vphi}_{\tau}+(\bar{\vphi}_{\tau}^{(\alpha)}-\bar{\vphi}_{\tau}))+c.c.$$
	together with the Leibniz rule and the Sobolev embedding on \eqref{gradestsplit4aid} to find as an upper bound of the fourth term in \eqref{gradestsplit}
	\begin{equation}\label{gradestsplit4}
		(\nu+\|(-\Delta)^{\frac\gamma4}(\vphi_{\tau}-\vphi_{\tau}^{(\alpha)})\|_2)\|(-\Delta)^{\frac\gamma4}(\vphi_{\tau}-\vphi_{\tau}^{(\alpha)})\|_2\nu.
	\end{equation}
	For the fifth term, we can basically repeat the same steps as for the fourth term to obtain
	\begin{equation}\label{gradestsplit5}
		(\nu+\|(-\Delta)^{\frac\gamma4}(\vphi_{\tau}-\vphi_{\tau}^{(\alpha)})\|_2)\|(-\Delta)^{\frac\gamma4}(\vphi_{\tau}-\vphi_{\tau}^{(\alpha)})\|_2^2.
	\end{equation}
	So, inserting \eqref{gradestsplit1}-\eqref{gradestsplit5} in \eqref{gradestsplit}, we obtain
	\begin{equation}
	\begin{split}
		\|(-\Delta)^{\frac{\gamma}{4}}(\vphi_t-\vphi_t^{(\alpha)})\|_2\lesssim\int_0^t\mathrm{d}\tau&\left\{\alpha^\varepsilon\e^{c\nu^2\tau}\|\vphi_0\|_{H^{(1+\varepsilon)\gamma/2}}^2\nu+\nu^2\|(-\Delta)^{\frac{\gamma}{4}}(\vphi_\tau-\vphi_\tau^{(\alpha)})\|_2\right.\\
		&\left.+\nu\|(-\Delta)^{\frac{\gamma}{4}}(\vphi_\tau-\vphi_\tau^{(\alpha)})\|_2^2+\|(-\Delta)^{\frac{\gamma}{4}}(\vphi_\tau-\vphi_\tau^{(\alpha)})\|_2^3\right\}.
	\end{split}
	\end{equation}
	Employing Gronwall's inequality yields the result.
	\par Next we estimate the $L^2$-distance of the solution $\vphi_t$ from its regularized version $\vphi_t^{(\alpha)}$. As in \cite{MS}, we start with the estimate
	\begin{equation}\label{l2estsplit}
	\begin{split}
		\frac12\partial_t\|\vphi_t-\vphi_t^{(\alpha)}\|_2^2&=-\partial_t\mathrm{Re}\scp{\vphi_t}{\vphi_t^{(\alpha)}}\\
		&=\lambda\mathrm{Im}\left\{\scp{\vphi_t}{\left(\frac{\alpha}{|\cdot|^\gamma(|\cdot|^\gamma+\alpha)}*|\vphi_t|^2\right)(\vphi_t^{(\alpha)}-\vphi_t)}\right.\\
		&\qquad+\left.\scp{\vphi_t}{\left(\frac1{|\cdot|^\gamma+\alpha}*(|\vphi_t|^2-|\vphi_t^{(\alpha)}|^2)\right)(\vphi_t^{(\alpha)}-\vphi_t)}\right\}\\
		&\lesssim\alpha\scp{|\vphi_t||\vphi_t^{(\alpha)}-\vphi_t|}{\frac{1}{|\cdot|^{2\gamma}}*|\vphi_t|^2}\\
		&\qquad+\scp{|\vphi_t||\vphi_t^{(\alpha)}-\vphi_t|}{\frac1{|\cdot|^\gamma}*||\vphi_t|^2-|\vphi_t^{(\alpha)}|^2|}.
	\end{split}
	\end{equation}
	In order to bound the first term, we use the Hardy-Littlewood-Sobolev inequality followed by H\"older's inequality again to get
	\begin{equation*}
		\begin{split}
			\alpha\scp{|\vphi_t||\vphi_t^{(\alpha)}-\vphi_t|}{\frac{1}{|\cdot|^{2\gamma}}*|\vphi_t|^2}&\lesssim\alpha\|\vphi_t(\vphi_t^{(\alpha)}-\vphi_t)\|_{\frac3{3-\gamma}}\|\vphi_t\|_{\frac6{3-\gamma}}^2\\
			&\lesssim\alpha\|\vphi_t^{(\alpha)}-\vphi_t\|_{\frac6{3-\gamma}}\|\vphi_t\|_{\frac6{3-\gamma}}^3\\
			&\lesssim\alpha \nu^3\|(-\Delta)^{\frac{\gamma}{4}}(\vphi_t^{(\alpha)}-\vphi_t)\|_2,
		\end{split}
	\end{equation*}
	where in the last step we used the Sobolev embedding. Using the above result, we obtain 
	\begin{equation}\label{l2estsplit1}
			\alpha\scp{|\vphi_t||\vphi_t^{(\alpha)}-\vphi_t|}{\frac{1}{|\cdot|^{2\gamma}}*|\vphi_t|^2}\lesssim C\alpha^{1+\varepsilon}
	\end{equation}
	for some constant $C=C_{\nu,\|\vphi_0\|_{H^{(1+\varepsilon)\gamma/2}},T}$. For the remaining term in \eqref{l2estsplit}, we apply the Hardy-Littlewood-Sobolev inequality followed by H\"older's inequality together with the triangle inequality to obtain
	\begin{equation*}
		\begin{split}
			\scp{|\vphi_t||\vphi_t^{(\alpha)}-\vphi_t|}{\frac1{|\cdot|^\gamma}*||\vphi_t|^2-|\vphi_t^{(\alpha)}|^2|}&\lesssim\|\vphi_t(\vphi_t^{(\alpha)}-\vphi_t)\|_{\frac6{6-\gamma}}\|(\vphi_t-\vphi_t^{(\alpha)})(\vphi_t+\vphi_t^{(\alpha)})\|_{\frac6{6-\gamma}}\\
			&\lesssim\|\vphi_t\|_{\frac6{3-\gamma}}(\|\vphi_t\|_{\frac6{3-\gamma}}+\|\vphi_t-\vphi_t^{(\alpha)}\|_{\frac6{3-\gamma}})\|\vphi_t-\vphi_t^{(\alpha)}\|_2^2\\
			&\lesssim\nu(\nu+\|(-\Delta)^{\frac{\gamma}{4}}(\vphi_t-\vphi_t^{(\alpha)})\|_2)\|\vphi_t-\vphi_t^{(\alpha)}\|_2^2,
		\end{split}
	\end{equation*}
	where in the last step we applied Sobolev's inequality. By the above, there is $C=C_{\nu,\|\vphi_0\|_{H^{(1+\varepsilon)\gamma/2,T}}}$ such that $\|(-\Delta)^{\frac{\gamma}{4}}(\vphi_t-\vphi_t^{(\alpha)})\|_2\leq C$. This yields
	\begin{equation}\label{l2estsplit2}
		\scp{|\vphi_t||\vphi_t^{(\alpha)}-\vphi_t|}{\frac1{|\cdot|^\gamma}*||\vphi_t|^2-|\vphi_t^{(\alpha)}|^2|}\leq C\|\vphi_t-\vphi_t^{(\alpha)}\|_2^2.
	\end{equation}
	Inserting \eqref{l2estsplit1} and \eqref{l2estsplit2} in \eqref{l2estsplit}, we arrive at
	\begin{equation*}
		\partial_t\|\vphi_t-\vphi_t^{(\alpha)}\|_2^2\lesssim C\left(\alpha^{1+\varepsilon}+\nu^2\|\vphi_t-\vphi_t^{(\alpha)}\|_2^2\right),
	\end{equation*}
	which, using Gronwall's inequality, gives the desired result.


\begin{thebibliography}{widestlabel}
	\bibitem[AFP]{AFP} Ammari, Z.; Falconi, M.; Pawilowski, B.: On the rate of convergence for the mean field approximation of many-body quantum dynamics. Preprint arXiv:1411.6284 (2014).
	\bibitem[AH]{AH} Anapolitanos, I.; Hott, M.: A simple proof of convergence to the Hartree dynamics in Sobolev trace norms. Journal of Mathematical Physics 57.12: 122108 (2016).
	\bibitem[AHH]{AHH} Anapolitanos, I.; Hott, M.; Hundertmark; D.: Derivation of the Hartree equation for compound Bose gases in the mean field limit. Reviews in Mathematical Physics 29.07: 1750022 (2017).
	\bibitem[AN]{AN} Ammari, Z.; Nier, F.: Mean field limit for bosons and infinite dimensional phase-space analysis. Annales Henri Poincaré. Vol. 9. No. 8.: 1503–1574 (2008).
	\bibitem[B]{B} Bose, S.N.: Plancks Gesetz und Lichtquantenhypothese. Zeitschrift für Physik. 26: 178–181 (1924).
	\bibitem[CH1]{CH1} Chen, X.; Holmer, J.: On the rigorous derivation of the 2D cubic nonlinear Schrödinger equation from 3D quantum many-body dynamics. Archive for Rational Mechanics and Analysis, 210(3), 909-954 (2013).
	\bibitem[CH2]{CH2} Chen, X.; Holmer, J.: Focusing quantum many-body dynamics: the rigorous derivation of the 1D focusing cubic nonlinear Schrödinger equation. Archive for Rational Mechanics and Analysis, 221(2), 631-676 (2016).
	\bibitem[CHKL]{CHKL} Cho, Y.; Hwang, G.; Kwon, S.; Lee, S.: On finite time blow-up for the mass-critical Hartree equations. Proceedings of the Royal Society of Edinburgh Section A: Mathematics, 145(3), 467-479 (2015).
	\bibitem[CHHO]{CHHO} Cho, Y.; Hajaiej, H.; Hwang, G.; Ozawa, T.: On the Cauchy problem of fractional Schrödinger equation with Hartree type nonlinearity. Funkcialaj Ekvacioj, 56(2), 193-224 (2013).
	\bibitem[CKS]{CKS} Cho, C. H.; Koh, Y.; Seo, I.: On inhomogeneous Strichartz estimates for fractional Schr\" odinger equations and their applications. arXiv preprint arXiv:1501.05399 (2015).
	\bibitem[CL]{CL} Cho, Y.; Lee, S.: Strichartz estimates in spherical coordinates. Indiana University Mathematics Journal, 991-1020 (2013). 
	\bibitem[CLS]{CLS} Chen, L.; Lee, J.O.; Schlein, B.: Rate of convergence towards Hartree dynamics. Journal of Statistical Physics, 144(4): 872-903 (2011).
	\bibitem[CLL]{CLL} Chen, L.; Lee, J.O.; Lee, J.: Rate of Convergence towards Hartree Dynamics with Singular Interaction Potential, arXiv:1708.07278 (2017)
	\bibitem[COX]{COX} Yonggeun C.; Ozawa, T.; Xia, S: Remarks on some dispersive estimates. Commun. Pure Appl. Anal., 10(4):1121–1128 (2011).
	\bibitem[CP1]{CP1} Chen, T.; Pavlović, N.: The quintic NLS as the mean field limit of a Boson gas with three-body interactions. J. Funct. Anal. 260(4), 959–997 (2011).
	\bibitem[CP2]{CP2} Chen, T.; Pavlović, N.: Derivation of the cubic NLS and Gross–Pitaevskii hierarchy from manybody dynamics in d= 3 based on spacetime norms. Annales Henri Poincaré. Vol. 15, No. 3, pp. 543-588. Springer Basel (2014).
	\bibitem[CW]{CW}Anderson, M. H.; Ensher, J. R.; Matthews, M. R., Wieman, C. E.; Cornell, E. A.: Observation of Bose-Einstein condensation in a dilute atomic vapor. science, 269(5221), 198-201 (1995).
	\bibitem[D]{D} Dinh, V.D.: On the Cauchy problem for the nonlinear semi-relativistic equation in Sobolev spaces (2017).
	\bibitem[DSS]{DSS} Deuchert, A.; Seiringer, R.; Yngvason, J.: Bose-Einstein Condensation in a Dilute, Trapped Gas at Positive Temperature. arXiv preprint arXiv:1803.05180 (2018).
	\bibitem[E]{E} Einstein, A.: Quantentheorie des einatomigen idealen Gases. Sitzungsberichte der Preussischen Akademie der Wissenschaften. 1: 3 (1925).
	\bibitem[ES]{ES} Erdös, L.; Schlein, B.: Quantum dynamics with mean field interactions: a new approach. Journal of Statistical Physics, 134(5-6), 859-870 (2009).
	\bibitem[ESY1]{ESY1}L. Erdös, B. Schlein, H.-T. Yau, Rigorous derivation of the Gross-Pitaevskii equation with a large
	interaction potential, J. Amer. Math. Soc. 22, 1099 (2009)
	\bibitem[ESY2]{ESY2}L. Erdös, B. Schlein, H.-T. Yau, Derivation of the Gross-Pitaevskii equation for the dynamics of
	Bose-Einstein condensate, Ann. of Math. 172, 291 (2010)
	\bibitem[EY]{EY} Erdös, L.; Yau, H.-T.: Derivation of the nonlinear Schrödinger equation from a many body
	Coulomb system. Adv. Theor. Math. Phys. 5, no. 6: 1169–1205 (2001).
	\bibitem[FGS]{FGS} Fröhlich, J.; Graffi, S.; Schwarz, S.: Mean-field- and classical limit of many-body Schrödinger dynamics for bosons. Comm. Math. Phys. 271, No. 3, 681-697 (2007).
	\bibitem[FKP]{FKP} Fröhlich, J.; Knowles, A.; Pizzo., A.: Atomism and quantization. J. Phys. A 40, no. 12, 3033-3045 (2007).
	\bibitem[FKS]{FKS} Fröhlich, J.; Knowles, A.; Schwarz, S.: On the mean-field limit of bosons with Coulomb two-body interaction. Communications in mathematical physics 288.3 (2009): 1023-1059.
	\bibitem[GM1]{GM1} Grillakis, M.; Machedon, M.: Pair excitations and the mean field approximation of interacting bosons, I. Communications in Mathematical Physics 324.2 (2013): 601-636.
	\bibitem[GM2]{GM2} Grillakis, M.; Machedon, M.: Pair excitations and the mean field approximation of interacting Bosons, II. Communications in Partial Differential Equations 42.1 (2017): 24-67.
	\bibitem[GV]{GV} Ginibre, J.; Velo, G.: The classical field limit of scattering theory for non-relativistic many-boson systems. I+II. Communications in Mathematical Physics 66.1: 37-76 (1979), and 68: 45–68 (1979).
	\bibitem[GW]{GW} Guo, Z.; Wang, Y.: Improved Strichartz estimates for a class of dispersive equations in the radial case and their applications to nonlinear Schrödinger and wave equations. Journal d'Analyse Mathématique, 124(1), 1-38 (2014).
 	\bibitem[GZ]{GZ} Guo, Q.; Zhu; S.: Sharp threshold of blow-up and scattering for the fractional Hartree equation. Journal of Differential Equations 264.4: 2802-2832 (2018).
 	\bibitem[Hei]{Hei}  Heil, T.: "Mean-field limits in bosonic systems." (Master's Thesis) LMU Munich (2012).
	\bibitem[Hep]{Hep} Hepp, K.: The classical limit for quantum mechanical correlation functions. Communications in Mathematical Physics 35.4: 265-277 (1974).
	\bibitem[Hw]{Hw} Hwang, G.: Almost sure local wellposedness of energy critical fractional Schrödinger equations with hartree nonlinearity. arXiv preprint arXiv:1504.06438 (2015).
	\bibitem[HS]{HS} Hong, Y.; Sire, Y.: On Fractional Schrödinger Equations in sobolev spaces. Communications on Pure and Applied Analysis 14.6: 2265-2282 (2015).
	\bibitem[K]{K}Davis, K. B.; Mewes, M. O.; Andrews, M. R.; Van Druten, N. J.; Durfee, D. S.; Kurn, D. M.; Ketterle, W.: Bose-Einstein condensation in a gas of sodium atoms. Physical review letters, 75(22), 3969 (1995).
	\bibitem[KM]{KM} Klainerman, S., Machedon, M.: On the uniqueness of solutions to the Gross–Pitaevskii hierarchy. Comm. Math. Phys. 279(1), 169–185 (2008).
	\bibitem[KP]{KP} Knowles, A.; Pickl, P.: Mean-Field Dynamics: Singular Potentials and Rate of Convergence. Communications in Mathematical Physics.  Volume 298, Issue 1:   101–138 (2010).
	\bibitem[KSS]{KSS} Kirkpatrick, K.; Schlein, B.; Staffilani, G.: Derivation of the two-dimensional nonlinear Schrödinger equation from many body quantum dynamics. American journal of mathematics, 133(1), 91-130 (2011).
	\bibitem[La1]{La1} Laskin, N.: Fractional quantum mechanics and Lévy path integrals, Phys. Lett. A 268, 298–304 (2000).
	\bibitem[La2]{La2} Laskin, N.: Fractional Schrödinger equation, Phys. Rev. E 66, 056108 (2002).
	\bibitem[LL]{LL} Lieb, E. H.; Loss, M.: Analysis, volume 14 of graduate studies in mathematics. American Mathematical Society, Providence, RI, 4 (2001).
	\bibitem[Lee]{Lee} Lee, J.P.: Rate of Convergence Towards Semi-Relativistic Hartree Dynamics (2013).
	\bibitem[Le]{Le} Lenzmann, E.: Well-posedness for semi-relativistic Hartree equations of critical type. Mathematical Physics, Analysis and Geometry 10.1: 43-64 (2007).
	\bibitem[Lu]{Lu} Lührmann, J.: Mean-field quantum dynamics with magnetic fields. Journal of Mathematical Physics 53.2: 022105 (2012).
	\bibitem[LSY]{LSY} Lieb, E. H.; Seiringer, R.; Yngvason, J.: Bosons in a trap: A rigorous derivation of the Gross-Pitaevskii energy functional. In The Stability of Matter: From Atoms to Stars. Springer, Berlin, Heidelberg: 685-697 (2001). 
	\bibitem[LSSY]{LSSY} Lieb, E. H.; Seiringer, R.; Solovej, J. P.; Yngvason, J.: The mathematics of the Bose gas and its condensation. Springer Science and Business Media (Vol. 34). (2005).
	\bibitem[LY]{LY} Lieb, E. H.; Yau, H. T.: The Chandrasekhar theory of stellar collapse as the limit of quantum mechanics. Communications in Mathematical Physics, 112(1), 147-174 (1987).
	\bibitem[M]{M} Mitrouskas, D.: Derivation of mean-field equations and next-order corrections for bosons and fermions (Doctoral dissertation, lmu) (2017).
	\bibitem[MO]{MO} Michelangeli, A., Olgiati, A.: Mean-field quantum dynamics for a mixture of Bose-Einstein condensates. 
	Anal.Math.Phys. 1-40 (2016).
	\bibitem[MPP]{MPP} Mitrouskas, D.; Petrat, S.; Pickl, P.: Bogoliubov corrections and trace norm convergence for the Hartree dynamics. arXiv preprint arXiv:1609.06264 (2016).
	\bibitem[MS]{MS} Michelangeli, A.; Schlein, B.: Dynamical collapse of boson stars. Communications in Mathematical Physics 311.3: 645-687 (2012).
	\bibitem[P]{P} Pickl, P.: A simple derivation of mean field limits for quantum systems. Letters in Mathematical Physics   97.2: 151-164 (2011).
	\bibitem[PS]{PS} Peng, C.; Shi, Q.: Stability of standing wave for the fractional nonlinear Schrödinger equation. Journal of Mathematical Physics, 59(1), 011508 (2018).
	\bibitem[RS]{RS}Rodnianski, I.; Schlein, B.: Quantum fluctuations and rate of convergence towards mean field dynamics. Communications in Mathematical Physics 291.1: 31-61 (2009).
	\bibitem[RS1]{RS1} Reed, M.;  Simon, B.: Methods of Modern Mathematical Physics. vol. I: Functional analysis. Academic press New York (1972).
	\bibitem[S]{S} Spohn, H.: Kinetic equations from Hamiltonian dynamics: Markovian limits. Reviews of Modern Physics 52.3: 569-615 (1980).
	\bibitem[T]{Tao} Tao, T.: Nonlinear dispersive equations: local and global analysis. No. 106. American Mathematical Soc., 2006.
	\bibitem[Zh]{Zh} Zhu, S.: On the blow-up solutions for the nonlinear fractional Schrödinger equation. Journal of Differential Equations, 261(2), 1506-1531  (2016).
	\bibitem[ZZ]{ZZ} Zhang, J.; Zhu, S.: Stability of standing waves for the nonlinear fractional Schrödinger equation. Journal of Dynamics and Differential Equations, 29(3), 1017-1030 (2017).
\end{thebibliography}
\end{document}